\documentclass[runningheads]{llncs}
\usepackage[symbol]{footmisc}
\usepackage[backend=bibtex, style=numeric, natbib=true]{biblatex}
\usepackage{algorithmic}
\usepackage[ruled,linesnumbered]{algorithm2e} 
\usepackage{amsmath,amsfonts,amsthm} 
\usepackage[titletoc]{appendix} 
\usepackage{booktabs} 
\usepackage{bm} 
\usepackage{caption} 
\usepackage{diagbox} 
\usepackage{float} 
\usepackage{graphicx} 
\graphicspath{{./figures/}} 

\usepackage[colorlinks=true, allcolors=blue]{hyperref} 

\usepackage{listings} 
\usepackage{mathrsfs} 
\usepackage{mathtools} 
\usepackage{multirow} 

\usepackage{subcaption} 
\usepackage{thmtools} 
\usepackage{thm-restate} 
\usepackage{tikz} 

\usepackage[textsize=normalsize]{todonotes} 
\usepackage{xcolor} 
\definecolor{pku-red}{RGB}{139,0,18} 
\usepackage[capitalize,noabbrev]{cleveref} 

\usepackage[top=1in, bottom=1in, left=1in, right=1in]{geometry} 

\crefname{algorithm}{Algorithm}{Algorithms}

\newcommand{\bbE}{\mathbb{E}}

\newcommand{\bbR}{\mathbb{R}}

\newcommand{\calB}{\mathcal{B}}

\newcommand{\calF}{\mathcal{F}}
\newcommand{\calG}{\mathcal{G}}

\newcommand{\calL}{\mathcal{L}}
\newcommand{\calM}{\mathcal{M}}
\newcommand{\calN}{\mathcal{N}}

\newcommand{\calP}{\mathcal{P}}

\newcommand{\calU}{\mathcal{U}}

\newcommand{\bmb}{{\bm{b}}}

\newcommand{\bmg}{{\bm{g}}}

\newcommand{\bmp}{{\bm{p}}}

\newcommand{\bmx}{{\bm{x}}}

\newcommand{\bmll}{\bm{\lambda}}

\newcommand{\iinn}{{i\in[n]}}

\newcommand{\jinm}{{j\in[m]}}

\newcommand{\iinM}{{i\in[M]}}

\newcommand{\ones}{\mathbf{1}} 
\newcommand{\zeros}{\mathbf{0}} 
\DeclareMathOperator*{\argmax}{arg\,max}

\newcommand{\rank}{\mathrm{rank}} 
\newcommand{\widesim}{{\scalebox{1.8}[1]{$\sim$}}} 
\newcommand{\iidd}{\overset{\iid}{\widesim}} 

\newcommand{\ie}{\emph{i.e.}}    
\newcommand{\eg}{\emph{e.g.}}    
\newcommand{\st}{\mathrm{s.t.}}   
\newcommand{\iid}{\mathrm{i.i.d.}}  

\newcommand{\OBJ}{\mathrm{OBJ}}
\newcommand{\OPT}{\mathrm{OPT}}
\newcommand{\LFW}{\mathrm{LFW}}
\newcommand{\LNW}{\mathrm{LNW}}

\newcommand{\GAP}{\mathrm{NG}}
\newcommand{\VoA}{\mathrm{VoA}}
\newcommand{\VoP}{\mathrm{VoP}}

\newcommand{\total}{\mathrm{total}}
\newcommand{\WSW}{\mathrm{WSW}}
\newcommand{\naive}{{na\"ive}}
\newcommand{\Naive}{{Na\"ive}}

\newcommand{\bmY}{\bm{Y}}
\newcommand{\bmlam}{\bm{\lambda}}

\usepackage[utf8]{inputenc}
\usepackage[T1]{fontenc}
\usepackage{graphicx}
\usepackage{color}

\begin{document}
\title{Large-Scale Contextual Market Equilibrium Computation through Deep Learning}
\titlerunning{Contextual Market Equilibrium Computation}



\author{
Yunxuan Ma\inst{1,2} \and
Yide Bian\inst{3} \and
Hao Xu\inst{4} \and
Weitao Yang\inst{4} \and
Jingshu Zhao\inst{4} \and
Zhijian Duan\inst{2} \and
Feng Wang\inst{4} \and
Xiaotie Deng\inst{2,5\dagger}
}
\authorrunning{Ma et al.}
%
\institute{
PKU-WUHAN Institute for Artificial Intelligence, Wuhan, China \and
CFCS, School of Computer Science, Peking University \\
\email{\{yunxuanma,zjduan,xiaotie\}@pku.edu.cn} \and
Yuanpei College, Peking University \\
\email{bian1d@stu.pku.edu.cn} \and
School of Computer Science, Wuhan University \\
\email{\{xuhao2002,yangweitao,candyzhao,fengwang\}@whu.edu.cn} \and
CMAR, Institute for AI, Peking University \\
$^{\dagger}$Corresponding author 
}

%
\maketitle              
\begin{abstract}

Market equilibrium is one of the most fundamental solution concepts in economics and social optimization analysis.
Existing works on market equilibrium computation primarily focus on settings with relatively few buyers.
Motivated by this, our paper investigates the computation of market equilibrium in scenarios with a large-scale buyer population, where buyers and goods are represented by their contexts.
Building on this realistic and generalized contextual market model, we introduce MarketFCNet, a deep learning-based method for approximating market equilibrium.
We start by parameterizing the allocation of each good to each buyer using a neural network, which depends solely on the context of the buyer and the good.
Next, we propose an efficient method to unbiasedly estimate the loss function of the training algorithm, enabling us to optimize the network parameters through gradient.
To evaluate the approximated solution, we propose a metric called Nash Gap, which quantifies the deviation of the given allocation and price pair from the market equilibrium.
Experimental results indicate that MarketFCNet delivers competitive performance and significantly lower running times compared to existing methods as the market scale expands, demonstrating the potential of deep learning-based methods to accelerate the approximation of large-scale contextual market equilibrium. 

\end{abstract}


\addtocounter{footnote}{-5}

\newpage
\section{Introduction}
\label{sec:intro}

Market equilibrium is a solution concept in microeconomics theory, which studies how \emph{individuals} amongst groups will exchange their \emph{goods} to get each one better off \citep{MWG:mas1995microeconomic}. 
The importance of market equilibrium is evidenced by the 1972 Nobel Prize awarded to John R. Hicks and Kenneth J. Arrow ``for their pioneering contributions to general economic equilibrium theory and welfare theory'' \citep{Nobel1972}.
Market equilibrium has wide application in fair allocation \citep{gao2021online}, as a few examples, fairly assigning course seats to students \citep{fair:course:budish2011combinatorial} or dividing estates, rent, fares, and others \citep{fair:estate:goldman2015spliddit}. Besides, market equilibrium is also considered for ad auctions with budget constraints where money has real value \citep{auction:PACE:conitzer2022pacing,auction:PACE:conitzer2022multiplicative}.


Existing works often use traditional optimization methods or online learning techniques to solve market equilibrium, which can tackle one market with around $400$ buyers and goods in experiments \citep{gao2020first, nan2023fast}. 
However, in realistic scenarios, there might be millions of buyers in one market (\eg\ job market, online shopping market). 
In these scenarios, the description complexity for the market is $\Omega(nm)$ and existing algorithms need $\Omega(nm)$ computational costs to do one iteration step on the market, if there are $n$ buyers and $m$ goods, which is unacceptable when $n$ is extremely large and potentially infinite. 

Contextual models come to the rescue.
The success of contextual auctions \citep{duan2022context,contextual-auctions:balseiro2023contextual} demonstrates the power of contextual models, in which each bidder and item can be described with low-dimensional representations, which consist of the information about bidders and items and further determines the value (or its distribution) of items to bidders.
In this way, auctions, as well as other economic problems, can be described more memory-efficiently, making it possible to accelerate the computation of these problems.
Inspired by the models of contextual auctions, we propose the concept of contextual markets in a similar way.
We argue that contextual markets can be useful to model large-scale markets aforementioned, since prior works have assumed the real market to be within some low-dimension space, and the values of goods to buyers are often not hard to speculate given the knowledge of goods and buyers \citep{context:kroer2019computing, context:kroer2019scalable}. 
Besides, contextual models never lose expressive power compared with raw models in the worst case \cite{contextual-expressive:bengio2009curriculum}.

This paper initiates the study of \emph{deep learning} for \emph{contextual} market equilibrium computation with a large (and potential infinite) number of buyers.
Following the framework of differentiable economics \citep{lottery-AMA:curry2022differentiable,differentiable-economy:dutting2023optimal, contract:wang2023deep}, we propose a deep-learning-based approach, \emph{MarketFCNet}, that takes the representations of one buyer and one good as input, and outputs the allocation of the good to the buyer.
The training on MarketFCNet targets at an unbiased estimator of the objective function in EG-convex program, which can be achieved with constant buyer samples.
Through learning an allocation function, we decrease the network iteration costs to $O(m)$ compared with $O(nm)$ in traditional buyer-wise methods, greatly accelerating the learning of market equilibrium.
In addition, the allocation of any item to any buyer, as well as the price of any item, can be accessed through one network call and thus have $O(1)$ query complexity.

In this way, we optimize the allocation function on ``buyer space'' implicitly, rather than optimizing the allocation to each buyer separately.
Therefore, the performance of MarketFCNet becomes independent of the number of buyers $n$. As a consequence, MarketFCNet shows its efficiency when $n$ becomes large.




The effectiveness of MarketFCNet is demonstrated by our experimental results.
As the market scale expands, MarketFCNet delivers competitive performance and significantly lower running times compared to existing methods in different experimental settings, demonstrating the potential of deep learning-based methods to accelerate the approximation of large-scale contextual market equilibrium.

To conclude, the contributions of this paper consist of three parts,
\begin{itemize}
    \item We propose a method, MarketFCNet, to approximate the contextual market equilibrium in which the number of buyers is large.
    \item We propose a measure, Nash Gap, to quantify the deviation of the given allocation and price pair from the market equilibrium.
    \item We conduct extensive experiments, demonstrating promising performance on the approximation measure and running time compared with existing methods.
\end{itemize}


\section{Related Works}
\label{sec:related-works}

\paragraph{Market Equilibriums}
The history of market equilibrium arises from microeconomics theory, where the concept of competitive equilibrium \citep[\S 10]{MWG:mas1995microeconomic} was proposed, and the existence of market equilibrium is guaranteed in a general setting \citep{equilibrium-exsistence:arrow1954existence,equilibrium-exsistence:walras2013elements}.
\citet{eisenberg1959consensus} first considered the linear market case, and proved that the solution of the EG-convex program constitutes a market equilibrium, which lays the polynomial-time algorithmic foundations for market equilibrium computation.
\citet{eisenberg1961aggregation} later showed that the EG program also works for a class of CCNH utility functions.
Shmyrev program later is also proposed to solve market equilibrium with linear utility with a perspective shift from allocation to price \citep{Shmyrev-program:shmyrev2009algorithm}, while \citet{Shmyrev-program:cole2017convex} later found that Shmyrev program is the dual problem of EG program with a change of variables.
There is also a branch of literature that consider computational perspective in more general settings such as indivisible goods \citep{equilibrium:complex:papadimitriou2001algorithms,equilibrium:complex:deng2002complexity,equilibrium:complex:deng2003complexity} and piece-wise linear utility \citep{equilibrium:piecewise-linear:vazirani2011market,equilibrium:piecewise-linear:garg2017settling,equilibrium:piecewise-linear:garg2022approximating}.

\paragraph{Algorithms of Solving Market Equilibriums}
There are abundant works that present algorithms to solve the market equilibrium and show the convergence results theoretically \citep{fisher:algo:cole2008fast}.
\citet{gao2020first} discusses the convergence rates of first-order algorithms for EG convex program under linear, quasi-linear, and Leontief utilities. 
\citet{nan2023fast} later designs stochastic optimization algorithms for the EG convex program and Shmyrev program with convergence guarantee and show some economic insight.
\citet{fisher:dist:algo:jalota2023fisher} proposes an ADMM algorithm for CCNH utilities and shows linear convergence results.
Besides, researchers are more engaged in designing dynamics that possess more economic insight.
For example, PACE dynamic \citep{gao2021online, PACE:liao2022nonstationary,PACE:yang2023greedy} and proportional response dynamic \citep{fisher:PR-dynamic:wu2007proportional,fisher:PR-dynamic:zhang2011proportional,fisher:PR-dynamic:cheung2018dynamics}, though the original idea of PACE arise from auction design \citep{auction:PACE:conitzer2022multiplicative, auction:PACE:conitzer2022pacing}.

\paragraph{Differential Economics}
With the fast growth of machine learning and neural networks, many existing works aim at resolving economic problems by deep learning approach, which falls into the differentiate economy framework \citep{differentiable-economy:dutting2023optimal}.
A mainstream is to approximate the optimal auction with differentiable models by neural networks \citep{auction-design:dutting2019optimal,auction-design:feng2018deep,auction-design:golowich2018deep,auction-design:rahme2021permutation}.
The problem of Nash equilibrium computation in normal form games \citep{nash-approximator:duan2023nash,nash-approximator:marris2022turbocharging,duan2023equivariant} and optimal contract design \citep{contract:wang2023deep} through deep learning also attracts researchers' attentions.
Among these methodologies, transformer architecture \citep{nash-approximator:marris2022turbocharging,duan2022context,auction-design:li2023learning} is widely used in solving economic problems for its permutation equivariance property.

\paragraph{Deep learning for market equilibrium}
To the best of our knowledge, no existing works try to approximate market equilibrium through deep learning. 
Besides, although some literature focuses on low-rank markets and representative markets \citep{context:kroer2019computing,context:kroer2019scalable},
our works firstly propose the concept of contextual market.
We believe that our approach will pioneer a promising direction for large-scale contextual market equilibrium computation.



\section{Contextual Market Modelling} 
\label{sec:preliminary}


In this section, we focus on the model of contextual market equilibrium in which goods are assumed to be divisible. 
Let the market consist of $n$ buyers, denoted as $1,...,n$, and $m$ goods, denoted as $1,...,m$. We denote $[k]$ as the abbreviation of the set $\{1,2,\dots,k\}$.
Each buyer $\iinn$ has a representation $b_i$, and each good $\jinm$ has a representation $g_j$. We assume that $b_i$ belongs to the buyer representation space $\calB$, and $g_j$ belongs to the good representation space $\calG$.
For a buyer with representation $b\in\calB$, she has a budget $B(b) > 0$.
Denote $Y(g)>0$ as the supply of good with representation $g$. Although many existing works \citep{gao2020first} assume that each good $j$ has \textit{unit} supply (i.e. $Y(g) \equiv 1$ for all $g\in\calG$) without loss of generality, their models can be easily generalized to our settings.

An \textit{allocation} is a matrix $\bmx = (x_{ij})_{\iinn,\jinm} \in \bbR^{n\times m}_+$, where $x_{ij}$ is the amount of good $j$ allocated to buyer $i$. We denote $\bmx_i = (x_{i1}, \dots, x_{im})$ as the vector of the bundle of goods that is allocated to buyer $i$.
The buyers' utility function is denoted as $u: \calB \times \bbR^m_+ \to \bbR_+$, here $u(b_i;\bmx_i)$ denotes the utility of buyer $i$ with representation $b_i$ when she chooses to buy $\bmx_i$.
We denote $u_i(\bmx_i)$ as an equivalent form of $u(b_i; \bmx_i)$ and often refer to them as the same thing. Similarly, $B(b_i), Y(g_j)$ and $B_i, Y_j$ are often referred to as the same thing, respectively.

Let $\bmp=(p_1,\dots,p_m) \in \bbR^m_+$ be the prices of the goods, the \emph{demand set} of buyer with representation $b_i$ is defined as the set of utility-maximizing allocations within budget constraint.
\begin{equation}
    D(b_i;\bmp) \coloneqq \argmax_{\bmx_i} \left\{u(b_i;\bmx_i)\mid  \bmx_i\in \bbR^m_+,\, \langle \bmp, \bmx_i\rangle \leq B(b_i) \right\}.
\end{equation}
A \emph{contextual market} is a 4-tuple: $\calM = \langle n, m, (b_i)_\iinn, (g_j)_\jinm \rangle$, where buyer utility $u(b_i;\bmx_i)$ is known given the information of the market.
We also assume budget function $B: \calB \to \bbR_+$ represents the budget of buyers and capacity function $Y: \calG \to \bbR_+$ represents the supply of goods. 
All of $u$, $B$, and $Y$ are assumed to be public knowledge and excluded from a market representation. 
This assumption mainly comes from two aspects: (1) these functions can be learned from historical data and (2) budgets and supplies can be either encoded in $b$ and $g$ in some way.

The \emph{market equilibrium} is represented as a pair $(\bmx, \bmp)$, $\bmx\in\bbR_+^{n\times m},\ \bmp\in \bbR_+^m$, which satisfies the following conditions.
\begin{itemize}
	\item \textit{Buyer optimality}: $\bmx_i \in D(b_i, \bmp)$ for all $\iinn$, 
	\item \textit{Market clearance}: $\sum_{i=1}^n x_{ij} \leq Y(g_j)$ for all $\jinm$, and equality must hold if $p_j > 0$. 
\end{itemize}

We say that $u_i$ is \textit{homogeneous} (with degree $1$) if it satisfies $u_i(\alpha \bmx_i) = \alpha u_i(\bmx_i)$ for any $\bmx_i \geq 0$ and $\alpha>0$ \citep[\S 6.2]{nisan2007algorithmic}.
Following existing works, we assume that $u_i$s are CCNH utilities, where CCNH represents concave, continuous, non-negative, and homogeneous functions\citep{gao2020first}. 
For CCNH utilities, a market equilibrium can be computed using the following \emph{Eisenberg-Gale convex program} (EG): 
\begin{equation}
\label{eq:eisenberg-gale-primal}
\max \sum_{i=1}^n B_i \log u_i(\bmx_i)\quad {\rm s.t.} \ \sum_{i=1}^n x_{ij} \leq Y_j, \ \bmx\geq 0.
\tag{EG}
\end{equation}

Theorem \ref{thm:GAOKROER} shows that the market equilibrium can be represented as the optimal solution of \eqref{eq:eisenberg-gale-primal}.

\begin{theorem}[\citet{gao2020first}]
\label{thm:eg-gives-me-for-certain-ui}
Let $u_i$ be concave, continuous, non-negative, and homogeneous (CCNH). Assume $u_i(\ones) >0$ for all $i$. Then, (i) \eqref{eq:eisenberg-gale-primal} has an optimal solution and (ii) any optimal solution $\bmx$ to \eqref{eq:eisenberg-gale-primal} together with its optimal Lagrangian multipliers $\bmp^*\in \bbR_+^m$ constitute a market equilibrium, up to arbitrary assignment of zero-price items. Furthermore, $\langle \bmp^*, \bmx^*_i \rangle = B_i$ for all $i$.
\label{thm:GAOKROER}
\end{theorem}
 
Based on \cref{thm:eg-gives-me-for-certain-ui}, it's easy to find that we can always assume $\sum_\iinn x_{ij} = Y_j$ while preserving the existence of market equilibrium, which states as follows.
\begin{proposition}
\label{prop:market-clear:equality}
Following the assumptions in \cref{thm:eg-gives-me-for-certain-ui}. For the following \emph{EG convex program with equality constraints},
\begin{equation}\label{eq:eisenberg-gale:equality}
\max \sum_{i=1}^n B_i \log u_i(\bmx_i)\quad {\rm s.t.} \ \sum_{i=1}^n x_{ij} = Y_j, \ \bmx\geq 0.
\end{equation}
Then, an optimal solution $\bmx^*$ together with its Lagrangian multipliers $\bmp^*\in \bbR^m_+$ constitute a market equilibrium.
Moreover, assume more that for each good $j$, there is some buyer $i$ such that $\frac{\partial u_i}{\partial x_{ij}}>0$ always hold whenever $u_i(\bmx_i) > 0$, then all prices are strictly positive in market equilibrium. As a consequence, \cref{eq:eisenberg-gale-primal} and \cref{eq:eisenberg-gale:equality} derive the same solution.
\end{proposition}

We leave all proofs to \cref{app:proofs}.
Since the additional assumption in \cref{prop:market-clear:equality} is fairly weak,
without further clarification, we always assume the conditions in \cref{prop:market-clear:equality} hold and the market clearance condition becomes
    $\sum_\iinn x_{ij} = Y(g_j),\ \forall \jinm$.



\section{MarketFCNet}
\label{sec:FCNet}

In this section, we introduce the MarketFCNet (denoted as Market Fully-Connected Network) approach to solve the contextual market equilibrium in the previous section.
The key point of MarketFCNet is to design an unbiased estimator of an optimization program whose solution coincides with the market equilibrium.
MarketFCNet does not rely on the number of buyers, making it an advantage to fit the infinite-buyer case without scaling on computational complexity.

\subsection{Problem Reformulation}
Following the idea of differentiable economics \citep{differentiable-economy:dutting2023optimal}, we utilize parameterized models to represent the allocation of good $j$ to buyer $i$, denoted as $x_\theta(b_i,g_j)$, and denote it as the allocation network, where $\theta$ is the network parameter.
Given buyer $i$ and good $j$, the network can automatically compute the allocation $x_{ij} = x_\theta(b_i,g_j)$.
The allocation to buyer $i$ is represented as $\bmx_i = \bmx_\theta(b_i,\bmg)$ and the allocation matrix is represented as $\bmx = \bmx_\theta(\bmb,\bmg)$.
Then the market clearance constraint can be reformulated as $\sum_\iinn x_\theta(b_i,g_j) = Y(g_j), \forall \jinm$ and the price constraint can be reformulated as $\bmx_\theta(\bmb,\bmg) \ge 0$.
Let $\calU(\calB)$ be the uniform distribution on the discrete set $\calB = \{b_i:\iinn\}$, then the EG program \eqref{eq:eisenberg-gale-primal} becomes,
\begin{equation}
\label{eq:eisenberg-gale:FC}
\begin{aligned}
    \max_{x_\theta} &\quad \OBJ(x_\theta) = \bbE_{b\sim \calU(\calB)} [B(b) \log u(b;\bmx_\theta(b,\bmg))]
    \\
    \st &\quad \bbE_{b\sim \calU(\calB)} [x_\theta(b,g_j)] = Y(g_j)/n, \forall \jinm
    \\
    &\quad \bmx_\theta(\bmb,\bmg) \ge 0
\end{aligned}\tag{EG-FC}
\end{equation}

For simplicity, we take $Y(g_j)/n \equiv 1$ for all $g_j$ and omit $\bbE_{b\sim \calU(\calB)}$ as $\bbE_{b}$ when the context is clear.
By this reformulation, the number of buyers, $n$, disappears from the program \eqref{eq:eisenberg-gale:FC}. As a consequence, MarketFCNet can potentially capture the scenarios even buyers are infinite or follow some distribution $\calF \in \Delta(\calB)$.

\subsection{Optimization}
The second constraint in \eqref{eq:eisenberg-gale:FC} can be hardcoded with network architecture (for example, a post-processed element-wise softplus function $\sigma(x) = \log(1 + \exp(x))$ that maps the set of real numbers to the set of positive numbers).
To address the first constraint, notice that the prices of goods are simply the Lagrangian multipliers of the first constraint in \eqref{eq:eisenberg-gale:FC}. 
Therefore, we employ the Augmented Lagrange Multiplier Method (ALMM) to explicitly extract the multipliers and solve the optimization problem \eqref{eq:eisenberg-gale:FC}.
We define $\calL_\rho(x_\theta, \bm{\lambda})$ as the Lagrangian where $\bm{\lambda}\in\bbR^m$ is the multipliers and $\rho$ is the quadratic penalty term.
The Lagrangian has the form as follows:
\begin{equation}
\begin{aligned}
    \calL_\rho(x_\theta;\bmlam) =& - \OBJ(x_\theta) + \sum_{j=1}^m \lambda_j \left(\bbE_b [ x_\theta(b,g_j) ] - 1\right)
    + \frac{\rho}{2} \sum_{j=1}^m \left(\bbE_b [ x_\theta(b,g_j) ] - 1\right)^2
\end{aligned}    
\end{equation}
Exactly computing the objective function seems intractable due to the extremely large and potentially infinite buyer size.
Therefore, we follow the framework in learning theory culture that many first-order algorithms can be applied (\eg, SGD, Adam) with only an unbiased gradient of the objective function \citep{SGD:amari1993backpropagation,SGD:bottou2010large}. 
Given such unbiased gradients as oracle, we provide our pseudo-codes in \cref{pseudo:FC}.
An illustration of the training procedure is also provided in \cref{fig:algorithm}.

\begin{figure*}[t]
    \centering
    \caption{Training process of MarketFCNet. On each iteration, the batch of $M$ independent buyers are drawn. each buyer and each good are represented as $k$-dimension context. The $(i,j)$'th element in the allocation matrix represents the allocation computed from $i$'th buyer and $j$'th good.
    MarketFCNet training process alternates between the training of allocation network and prices. The training of allocation network needs to achieve an unbiased estimator $\widehat{\calL}_\rho(x_\theta;\lambda)$ of the loss function $\calL_\rho(x_\theta;\lambda)$, followed by gradient descent. The training of prices need to get an unbiased estimator $\widehat{\Delta}\lambda_j$ of $\Delta \lambda_j$, followed by ALMM updating rule $\lambda_j \leftarrow \lambda_j + \beta_t \widehat{\Delta}\lambda_j$.}
    \includegraphics[width=0.9\textwidth]{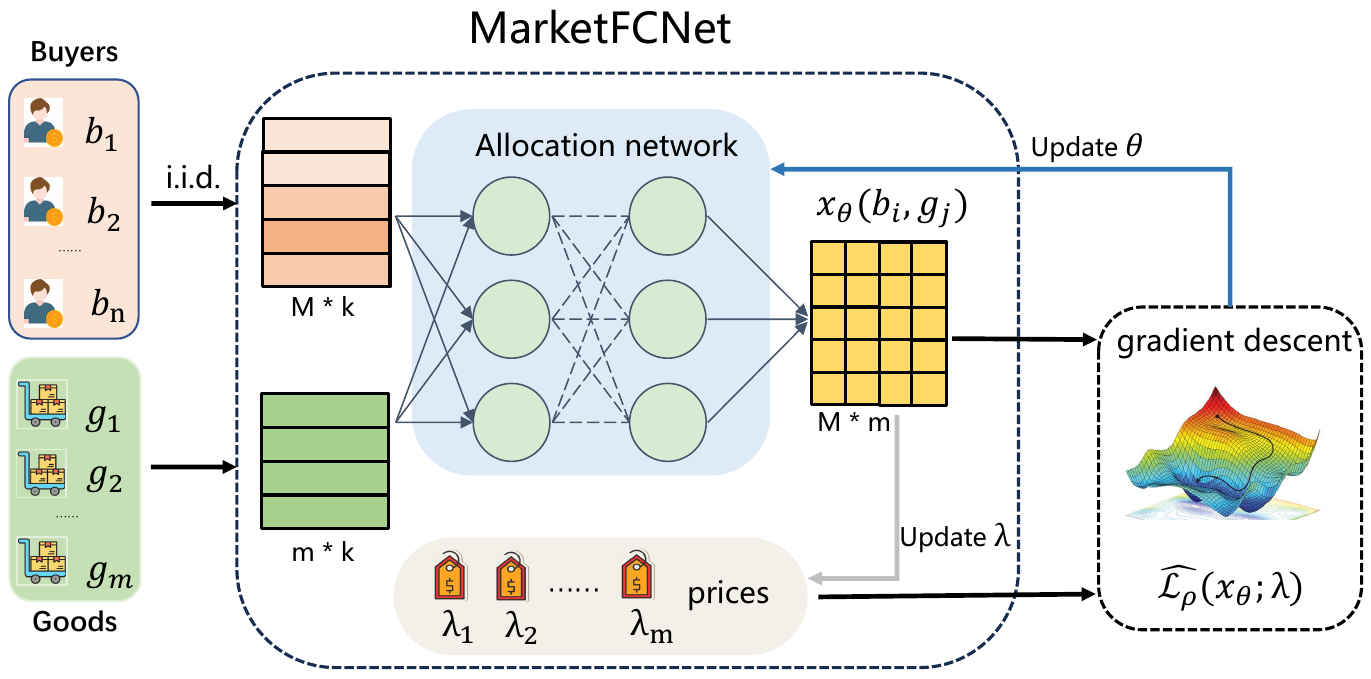}
    \label{fig:algorithm}
\end{figure*}

\begin{algorithm}[t]
\caption{MarketFCNet}
\label{pseudo:FC}
\KwIn{An oracle of buyer sampler $b \sim \calU(\calB)$, goods $g_1,...,g_m$, batch size $M_1$ and $M_2$, quadratic penalty term $\rho$, iteration $K$ for optimizing Lagrangian, step size $(\beta_t)_{t=1}^\infty$ for optimizing multipliers.}
\KwOut{Allocation network $x_\theta(b,g)$, price $p_j$ for each good $g_j$.}
Initialize an allocation network $x_\theta(b,g)$ and multipliers $\{\lambda_j\}_\jinm$.\\
\For{$t = 1,2,...$ until converged}
{
    \For{$k = 1,2,... K$}
    {
        Get an unbiased estimator $\hat{\calL}_\rho(x_\theta;\bmlam)$ with batch size $M_1$ such that $ \bbE[\hat{\calL}_\rho(x_\theta;\bmlam)] = \calL_\rho(x_\theta;\bmlam)$.\\
        Call the optimizer to optimize $\theta$ on $\hat{\calL}_\rho(x_\theta;\bmlam)$ for one step.\\
    }
    \For{$j=1,...,m$}
    {
        Get an unbiased estimator $\hat{\Delta} \lambda_j$ with batch size $M_2$ such that $\bbE[\hat{\Delta} \lambda_j] = \Delta \lambda_j \coloneqq \rho (\bbE_{b}[x_\theta(b,g_j)] - 1 )$.\\
        Updates $\lambda_j$ with $\lambda_j \leftarrow \lambda_j + \beta_t \Delta \lambda_j $.\\
    }
}
\textbf{return} $x_\theta(b,g)$,$\{\lambda_j\}_\jinm$.
\end{algorithm}

Let $X(w)$ be a deterministic function where $w$ is a (potentially high-dimensional) random variable with distribution $\calP(w)$ and $x$ be a real number. We call $X(w)$ is an unbiased estimator of $x$ if $\bbE_{w\sim \calP(w)}[X(w)] = x$.

From \cref{pseudo:FC}, it's clear that in order to finish the ALMM algorithm, we need to obtain unbiased estimators of the following two expressions.
\begin{itemize}
    \item An unbiased estimator of $ \calL_\rho(x_\theta;\bmlam)$.
    \item An unbiased estimator of $\Delta \lambda_j$, where $\Delta \lambda_j$ is given by 
    $\textstyle \Delta \lambda_j = \rho \left( \mathbb{E}_b [x_\theta(b, g_j)] - 1 \right)$.
\end{itemize}

The next section introduces how to get unbiased estimators for these expressions.

\subsection{Extract Unbiased Estimators}

\paragraph{Unbiased estimator of $\Delta \lambda_j$.}

To get an unbiased estimator of $\Delta \lambda_j$, it's enough to obtain an unbiased estimator of $\bbE_b[x_\theta(b,g_j)]$. 
By applying Monte Carlo method, we can choose batch size $M$ and sample $b_1, b_2,...,b_M \sim \calU(\calB)$, then 
$\frac{1}{M} \sum_{i=1}^M x_\theta(b_i, g_j)$ forms an unbiased estimator.

Define $\hat{\Delta}_j(b_1,\dots, b_M) = \rho (\frac{1}{M}\sum_{i=1}^M x_\theta(b_i, g_j) - 1)$. It's clear to see that 
$\bbE_{b_i \iidd \calU(\calB)}[\hat{\Delta}_j(b_1,\dots, b_M)] = \Delta_j$.

\paragraph{Unbiased estimator of $\calL_p(x_\theta;\bmlam)$.}

Notice that $\calL_p(x_\theta;\bmlam)$ is the sum of three terms: the objective term $\OBJ(x_\theta)$, the multiplier term $\sum_{j=1}^m \lambda_j \left(\bbE_b [ x_\theta(b,g_j) ] - 1\right)$ as well as the quadratic term $\frac{\rho}{2} \sum_{j=1}^m $ $\left(\bbE_b [ x_\theta(b,g_j) ] - 1\right)^2$.
We only need to get an unbiased estimator for each term.
For the first and the second term, the technique to achieve an unbiased estimator is similar to those in $\Delta \lambda_j$ and thus omitted. 

For the last term, the quadratic dependency on expectation requires additional attention. Notice that
\begin{align}
    & \left( \bbE_b \left[x_\theta(b,g_j) \right] - 1\right)^2
    = \left( \bbE_b \left[x_\theta(b,g_j) \right] - 1\right) \cdot \left( \bbE_{b'} \left[x_\theta(b',g_j) \right] - 1\right)
\end{align}
Therefore, we can sample $b_1,...,b_M, b'_1,...,b'_M \sim U(\calB)$ and compute
\begin{align}
    \frac{\rho}{2} \cdot \frac{1}{M} \sum_{i=1}^M \sum_{j=1}^m \left( x_\theta(b_i,g_j) - 1\right) \cdot \left( x_\theta(b'_i,g_j) - 1\right)
\end{align}
which provides an unbiased estimator for the last term.

Formally, we formulate the unbiased estimator in following proposition.

\begin{proposition}[Unbiased estimator of $\calL_p(x_\theta;\bmlam)$]
\label{prop:unbiased}
Define 
\begin{align*}
    &\hat{\calL}_\rho(x_\theta;\bm{\lambda};b^1,...,b^{2M}) = - \frac{1}{M} \sum_\iinM \left[ B(b^i) \log u(b^i;\bmx_\theta(b^i,\bmg)) \right]
    \\
    +& \sum_\jinm \lambda_j \left( \frac{1}{M} \sum_\iinM x_\theta(b^i,g_j) - 1 \right)
    + \frac{\rho}{2M}\sum_\jinm \sum_\iinM \left( x_\theta(b^i,g_j) - 1 \right) \left( x_\theta(b^{i+M},g_j) - 1 \right).
\end{align*}

Then, $\hat{\calL}_\rho(x_\theta;\bm{\lambda};b^1,...,b^{2M})$ is an unbiased estimator of $\calL_\rho(x_\theta;\bm{\lambda})$, given the distribution ${b^k}_{1\le k\le 2M} \iidd \calU(\calB)$.
\end{proposition}

\section{Performance Measures of Market Equilibrium}
\label{sec:performance}

In this section, we propose \emph{Nash Gap} to measure the performance of an approximated market equilibrium and show that Nash Gap preserves the economic interpretation for market equilibrium.
To introduce Nash Gap, we first introduce two types of welfare, Log Nash Welfare and Log Fixed-price Welfare in \cref{def:LNW} and \cref{def:LFW}, respectively.

\begin{definition}[Log Nash Welfare]
\label{def:LNW}
The Log Nash Welfare (abbreviated as $\LNW$) is defined as
\begin{align}
    \LNW(\bmx) = \frac{1}{B_\total} \sum_\iinn B_i \log u_i(\bmx_i),
\end{align}
where $B_\total = \sum_\iinn B_i$ is the total budget for buyers.
\end{definition}

Log Nash Welfare is fundamentally the logarithm of Nash Welfare with an economic interpretation of the geometric mean of buyers' utilities, which has been widely studied in literature \citep{NashWel:banerjee2022online, NashWel:huang2023online}.


\begin{definition}[Fixed-price utility and Log Fixed-price Welfare]
\label{def:LFW}
We define the fixed-price utility for buyer $i$ as,
\begin{align}
    \tilde{u}(b_i;\bmp) = \max_{\bmx_i} \{ u(b_i;\bmx_i) \mid \bmx_i\in\bbR^m_+, \langle \bmp ,\bmx_i \rangle \le B(b_i) \} 
\end{align}
which represents the optimal utility that buyer $i$ can obtain at the price level $\bmp$, regardless of the market clearance constraints.
The Log Fixed-price Welfare (abbreviated as $\LFW$) is defined as the weighted sum of the logarithm of Fixed-price utility, 
\begin{align}
    \mathrm{LFW}(\bmp) = \frac{1}{B_\total} \sum_\iinn B_i \log \tilde{u}_i(\bmp)
\end{align}
\end{definition}

Based on these definitions, we present the definition of Nash Gap.
\begin{definition}[Nash Gap]
    We define Nash Gap (abbreviated as NG) as the difference between Log Nash Welfare and Log Fixed-price Welfare, \ie
    \begin{align}
    \GAP(\bmx,\bmp) = \LFW(\bmp) - \LNW(\bmx)
    \end{align}
\end{definition}

\subsection{Properties of Nash Gap}

To show why $\GAP$ is useful in the measure of market equilibrium, we first observe below statement, 
\begin{proposition}[Price constraints]
\label{prop:price-constraint}
If $(\bmx, \bmp)$ constitutes a market equilibrium, the following identity always holds,
\begin{align}
    \sum_\jinm p_j Y_j = \sum_\iinn B_i
\end{align}
\end{proposition}
Below, we state the most important theorem in this paper.
\begin{theorem}
\label{thm:gap-largerthan0}
Let $(\bmx,\bmp)$ be a pair of allocation and price. Assuming the allocation satisfies market clearance and the price meets price constraint, then we have $\GAP(\bmx,\bmp) \ge 0$.

Moreover, $\GAP(\bmx,\bmp) = 0$ if and only if $(\bmx,\bmp)$ is a market equilibrium.
\end{theorem}

\cref{thm:gap-largerthan0} shows that Nash Gap is an ideal measure of the solution concept of market equilibrium since it holds the following properties,
\begin{itemize}
    \item $\GAP(\bmx,\bmp)$ is continuous on the inputs $(\bmx,\bmp)$.
    \item $\GAP(\bmx,\bmp)\ge 0$ always hold. (under conditions in \cref{thm:gap-largerthan0})
    \item $\GAP(\bmx,\bmp)=0$ if and only if $(\bmx,\bmp)$ meets the solution concept.
    \item The computation of $\GAP$ does not require the knowledge of an equilibrium point $(\bmx^*,\bmp^*)$
\end{itemize}

Since some may argue that $\GAP(\bmx,\bmp)$ is not intuitive to understand, we consider some more intuitive measures, the Euclidean distance to the market equilibrium, \ie, $||\bmx - \bmx^*||$ and $||\bmp - \bmp^*||$, as well as the difference on Weighted Social Welfare, $|\WSW(\bmx) - \WSW(\bmx^*)|$, where $\WSW(\bmx) \coloneqq \sum_\iinn \frac{B_i}{B_\total} u_i(\bmx_i)$, and show the connection between $\GAP$ and these intuitive measures.

\begin{proposition}[Informal]
\label{prop:gap-epsilon}
Under some technical assumptions 
, if $\GAP(\bmx,\bmp) = \varepsilon$, we have:
\begin{itemize}
    \item $||\bmp - \bmp^*|| = O(\sqrt{\varepsilon})$.
    \item $||\bmx_i - \bmx_i^*|| = O(\sqrt{\varepsilon})$ for all $i$.
    \item $|\WSW(\bmx) - \WSW(\bmx^*)| = O(\varepsilon)$.
\end{itemize}
\end{proposition}

\cref{prop:gap-epsilon} means that intuitive measures (\eg, Euclidean distance to the equilibrium point) are upper bounded by Nash Gap with a square root rate, serving as a certificate that Nash Gap is an ideal measure.

Finally, we give a saddle-point explanation for Nash Gap.
\begin{corollary}
\label{cor:market-equilibrium-is-saddle-point}
Within market clearance and price constraint, we have
\begin{align}
    \min_\bmp \LFW(\bmp) = \max_\bmx \LNW(\bmx)
\end{align}
\end{corollary}
\cref{cor:market-equilibrium-is-saddle-point} provides an economic interpretation for GAP. Market equilibrium can be seen as the saddle point over social welfare, and the social welfare for $\bmx$ can be actually implemented while the social welfare for $\bmp$ is virtual and desired by buyers. Nash Gap measures the gap between the ``desired welfare'' and the ``implemented welfare'' for buyers.

\subsection{Measures in General Cases}

Since $\GAP$ only works for $(\bmx,\bmp)$ that satisfies market clearance and price constraints, we generalize the measure of $\GAP$ to the full outcome space in this section, which drives us to design some reasonable measures that can measure the performance of all positive $(\bmx,\bmp)$.

We first notice that any equilibrium must satisfy the conditions of \emph{market clearance} and \emph{price constraint}, we first project arbitrary positive $(\bmx,\bmp)$ to the space where these constraints hold. 
Specifically, if we let
\begin{align}
    \alpha_j =& \frac{Y_j}{\sum_i x_{ij}},\quad \tilde{x}_{ij} = x_{ij} \cdot \alpha_j
    \qquad& 
    \beta = \frac{\sum_i B_i}{\sum_j Y_j p_j},\quad \tilde{p}_j = \beta \cdot p_j
\end{align}
then $(\tilde{\bmx}, \tilde{\bmp})$ satisfies these constraints and we consider $\GAP(\tilde{\bmx},\tilde{\bmp})$ as the equilibrium measure.

Besides, we also need to measure how far is the point $(\bmx,\bmp)$ to the space within the conditions of \emph{market clearance} and \emph{price constraint}.
we propose the following two measurements, called Violation of Allocation (abbreviated as $\VoA$) and Violation of Price (abbreviated as $\VoP$), respectively.
\begin{align}
    \VoA(\bmx) \coloneqq& \frac{1}{m} \sum_j |\log \alpha_j|,
    \qquad
    \VoP(\bmp) \coloneqq |\log \beta|
\end{align}
From the expressions of $\VoA$ and $\VoP$, we know that these two constraints hold if and only if $\VoA(\bmx) = 0$ and $\VoP(\bmp) = 0$.


We argue that this projection is of economic meaning. If $(\bmx,\bmp)$ constitutes a market equilibrium and we scale the budget with a factor of $\beta$, then $(\bmx,\beta\bmp)$ constitutes a market equilibrium in the new market. Similarly, if we scale the value for each buyer with factor $\bm{\alpha}^{-1}$ (here $\bm{\alpha}$ is a vector in $\bbR^m_+$ and $\bm{\alpha}^{-1}$ is element-wise inverse of $\bm{\alpha}$) and capacity with factor $\bm{\alpha}$, then, $(\bm{\alpha}\cdot \bmx, \frac{\bmp}{\bm{\alpha}})$ constitute a market equilibrium in the new market. 
These instances are evidence that market equilibrium holds a linear structure over market parameters. Therefore, a linear projection can eliminate the effect from linear scaling, while preserving the effect from orthogonal errors.

Notice that $\bmx = \tilde{\bmx}$ and $\bmp = \tilde{\bmp}$ if and only if $\VoA(\bmx)=0$ and $\VoP(\bmp) = 0$, respectively.
From \cref{thm:gap-largerthan0} We can easy derive following statement:
\begin{proposition}
\label{prop:measure-marketequilibrium}
For arbitrary $\bmx \in \bbR^{n\times m}_+,\bmp\in\bbR^m_+$, we have $\VoA(\bmx)\ge 0, \VoP(\bmp)\ge 0, \GAP(\tilde{\bmx}, \tilde{\bmp})\ge 0$ always hold. 
Moreover, $(\bmx, \bmp)$ is a market equilibrium if and only if $\VoA(\bmx) = \VoP(\bmp) = \GAP(\tilde{\bmx},\tilde{\bmp}) = 0$.
\end{proposition}

\cref{prop:measure-marketequilibrium} is a certificate that $\VoA(\bmx), \VoP(\bmp), \GAP(\tilde{\bmx}, \Tilde{\bmp})$ together form a good measure for market equilibrium. Therefore, in our experiments, we compute these measures of solutions and prefer a lower measure without further clarification.








\section{Experiments}
\label{sec:experiments}

In this section, we present empirical experiments that evaluate the effectiveness of MarketFCNet.
Though briefly mentioned in this section, we leave the details of baselines, implementations, hyper-parameters, and experimental environments to
\cref{app:experiments}.


\subsection{Experimental Settings}

In our experiments, all utilities are chosen as CES utilities, which satisfies CCNH conditions and captures a wide utility class including linear utilities, Cobb-Douglas utilities, and Leontief utilities
\citep{CES:varian1992microeconomic, CES-general:arrow1961capital}.
CES utilities have the form of
$u_i(x_i) = \left(\sum_\jinm v^\alpha_{ij} x^{\alpha}_{ij} \right)^{1/\alpha}$
with $\alpha \le 1$. The fixed-price utilities for CES utility, which are necessary to measure the performance, are derived
in \cref{app:derivation}.

In order to evaluate the performance of MarketFCNet, we compare them mainly with a baseline that directly maximizes the objective in EG convex program with gradient ascent algorithm (abbreviated as \emph{EG}), which is widely used in the field of market equilibrium computation. 
Besides, we also consider a momentum version of \emph{EG} algorithm with momentum $\beta=0.9$ (abbreviated as \emph{EG-m}).



We also consider a \naive\ allocation and pricing rule (abbreviated as \emph{\Naive}), which can be regarded as the benchmark of the experiments:
\begin{align}
    x_{ij} =\ 1,\quad p_j = \frac{\sum_\iinn B_i}{m V_j},\quad \text{for all $i,j$}
\end{align}

In \emph{\Naive}, the allocation is evenly distributed such that the market clearance holds. The price for each good is designated such that the price constraints also hold.

In the following experiments, MarketFCNet is abbreviated as \emph{FC}.
Notice that \emph{\Naive} always gives an allocation that satisfies market clearance and price constraints, while \emph{EG}, \emph{EG-m} and \emph{FC} do not.

\subsection{Experiment Results}

\paragraph{Comparing with Baselines}
\begin{table}[t]
    \centering
    \caption{Comparison of MarketFCNet with baselines: $n=1,048,576$ buyers and $m=10$ goods. The GPU time for MarketFCNet represents the training time and testing time, respectively.}
    \begin{tabular}{lrrrr}
    \toprule
    Methods & $\GAP$ & VoA & VoP & GPU Time \\
    \midrule\midrule
    \Naive & 3.65e-1 & 0 & 0 & 3.57e-3 \\
    \midrule
    EG & 2.17e-2 & 2.620e-1 & 7.031e-2 & 197 \\
    \midrule
    EG-m & \textbf{2.49e-4} & 6.01e-2 & 9.77e-2 & 100 \\
    \midrule
    FC & 1.63e-3 & \textbf{1.416e-2} & \textbf{6.750e-3} & \textbf{43.6}; \textbf{9.63e-2} \\
    \bottomrule
    \end{tabular}
    \label{tab:exp:FC:basic}
\end{table}

We choose the number of buyers $n=1,048,576=2^{20}$, number of items $m=10$, CES utility parameter $\alpha=0.5$ and representation with standard normal distribution as the basic experimental environment of MarketFCNet;
We consider $\GAP(\tilde{\bmx},\tilde{\bmp}), \VoA(\bmx), \VoP(\bmp)$ and the running time of algorithms as the measures.
Without special specifications, these parameters are default settings among other experiments.
Results are presented in \cref{tab:exp:FC:basic}.
From these results, we can see that the approximations of MarketFCNet are competitive with \emph{EG} and \emph{EG-m} and far better than \Naive, which means that the solution of MarketFCNet is very close to market equilibrium.
MarketFCNet also achieves a much lower running time compared with \emph{EG} and \emph{EG-m}, which indicates that these methods are more suitable for large-scale market equilibrium computation.
In the following experiments, $\VoA$ and $\VoP$ measures are omitted and we only report $\GAP$ and running time.

\paragraph{Experiments in different parameters settings}
\begin{figure*}[t]
    \centering
    
    \begin{subfigure}{0.24\textwidth}
    \includegraphics[width=\textwidth]{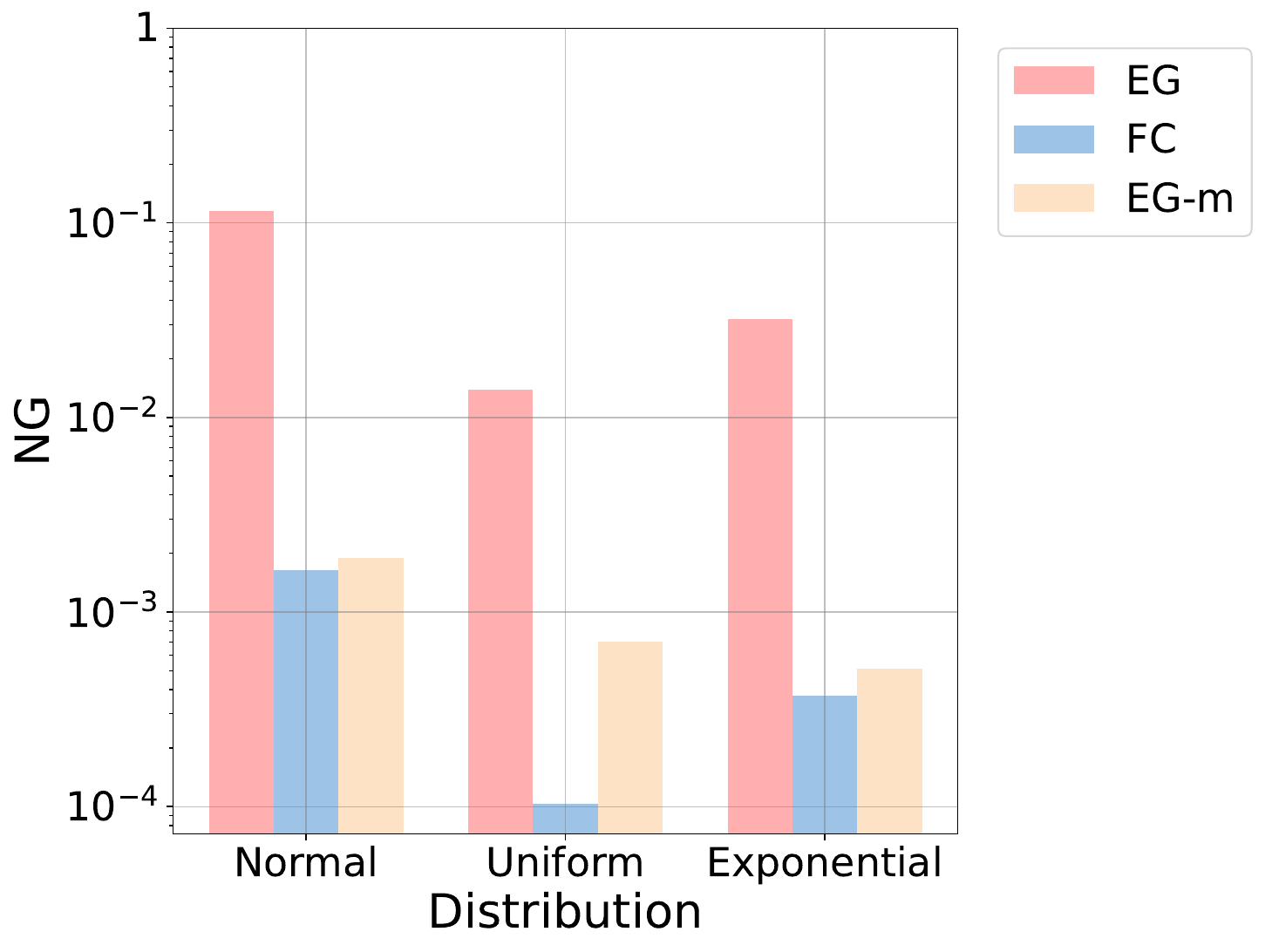}
    \caption{Nash Gap on different distributions.}
    \label{fig:exp:FC:dist-GAP}
    \end{subfigure}
    \hfill
    \begin{subfigure}{0.24\textwidth}
    \includegraphics[width=\textwidth]{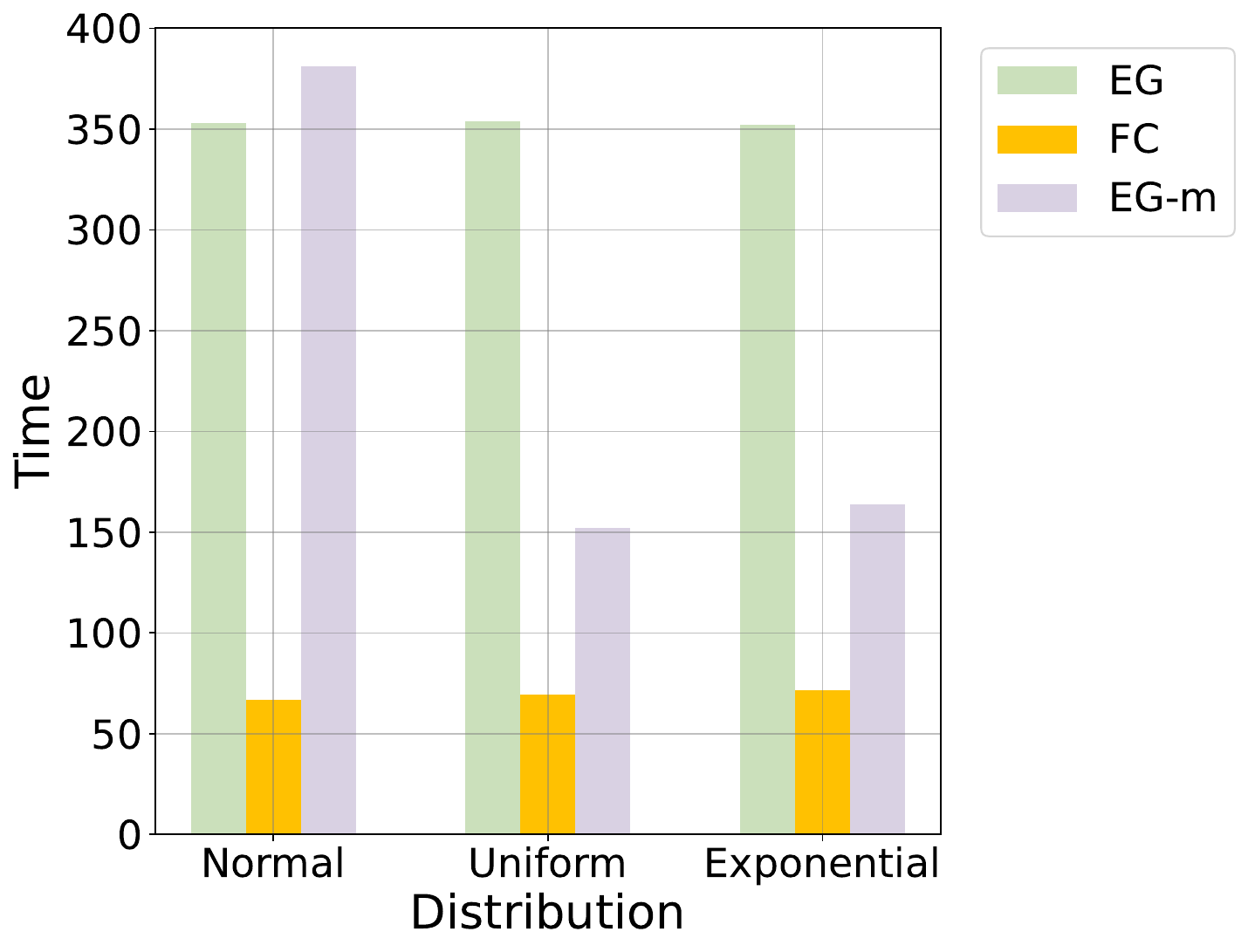}
    \caption{GPU time on different distributions.}
    \label{fig:exp:FC:dist-time}
    \end{subfigure}
    \hfill
    \begin{subfigure}{0.24\textwidth}
    \includegraphics[width=\textwidth]{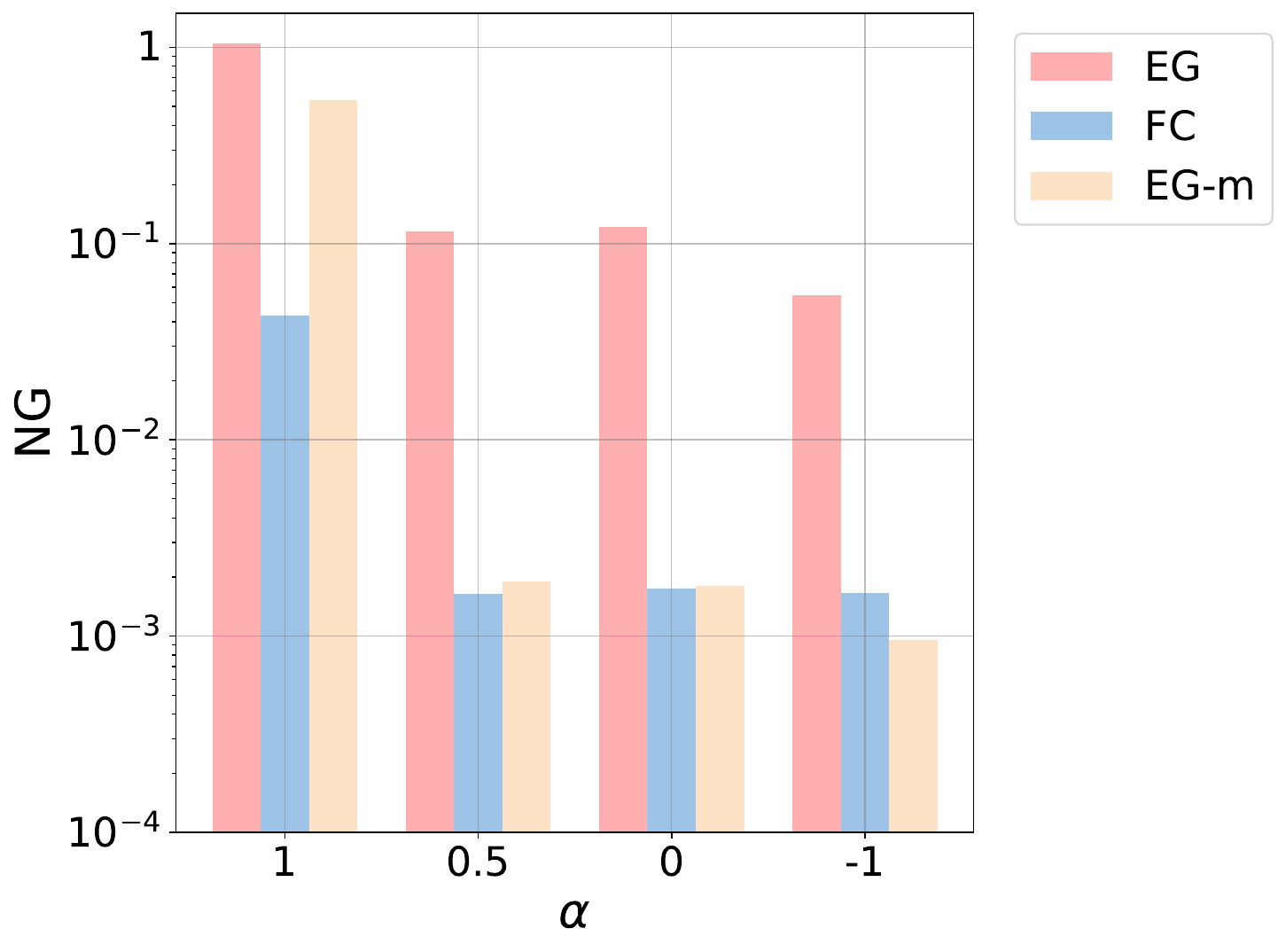}
    \caption{Nash Gap on different $\alpha$.}
    \label{fig:exp:FC:alpha-GAP}
    \end{subfigure}
    \hfill
    \begin{subfigure}{0.24\textwidth}
    \includegraphics[width=\textwidth]{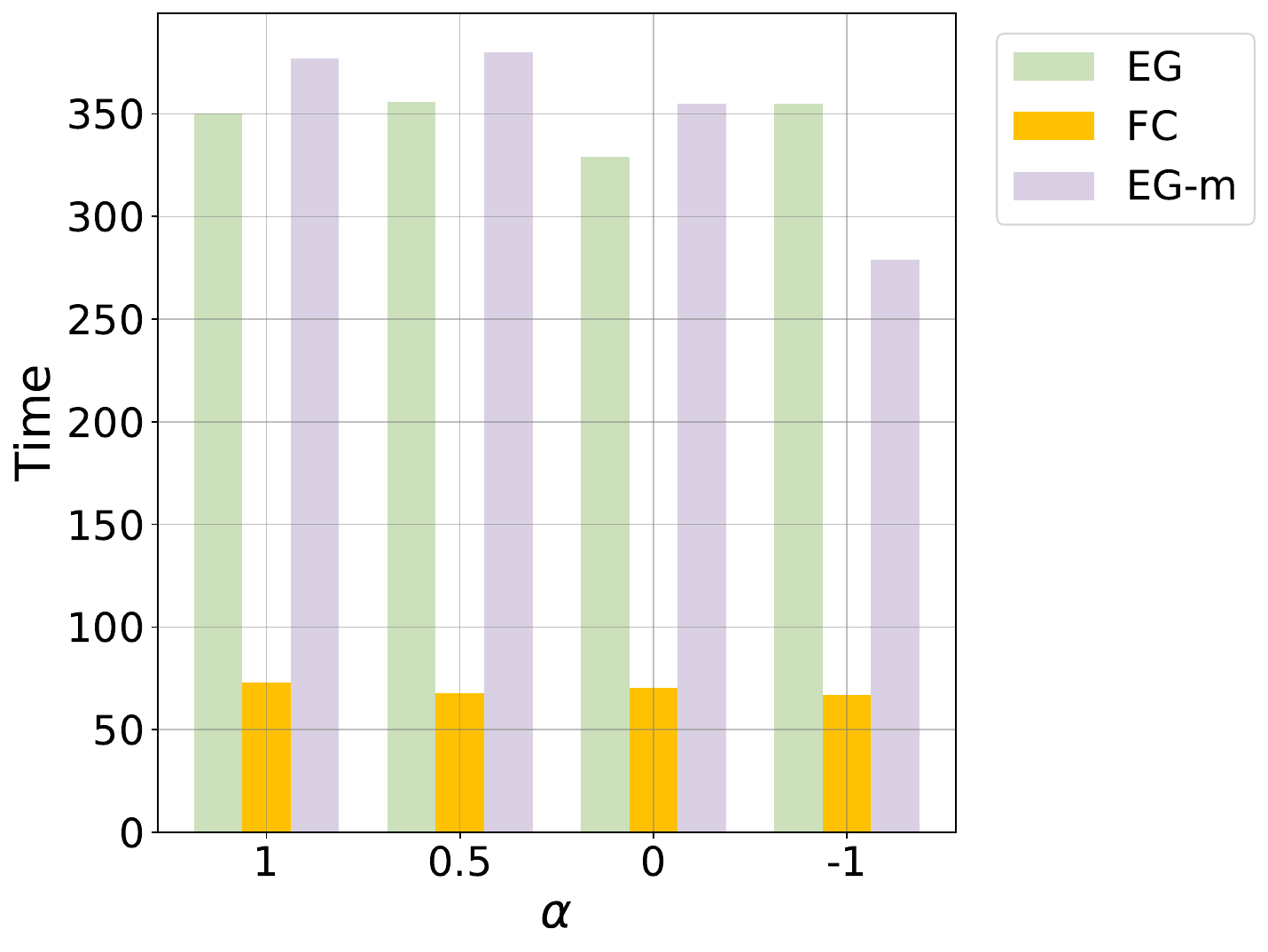}
    \caption{GPU time on different $\alpha$.}
    \label{fig:exp:FC:alpha-time}
    \end{subfigure}
    
    \caption{The Nash Gap and GPU running time for different approaches: MarketFCNet, EG, and EG-m. Different colors represent different approaches. Market size is fixed with $n=4, 194, 304$ buyers and $m=10$ goods.}
    \label{fig:exp:FC:dist-and-alpha}
\end{figure*}

In this experiment, the market scale is chosen as $n=4, 194, 304$ and $m=10$.
We consider experiments on different distributions of representation, including normal distribution, uniform distribution, and exponential distribution. See (a) and (b) in \cref{fig:exp:FC:dist-and-alpha} for results. 
We also consider different $\alpha$ in our experimental settings.
Specifically, our settings consist of: 1) $\alpha=1$, the utility functions are linear; 
2) $\alpha=0.5$, where goods are substitutes; 
3) $\alpha=0$, where goods are neither substitutes nor complements; 
4) $\alpha=-1$, where goods are complements.
More detailed results are shown in (c) and (d) \cref{fig:exp:FC:dist-and-alpha}. The performance of MarketFCNet is robust in both settings.
From experimental results we see that MarketFCNet performs well on different settings.



\paragraph{Different market scale for MarketFCNet}
\begin{figure*}[t]
    \centering
    \begin{subfigure}{0.3\textwidth}
    \includegraphics[width=\textwidth]{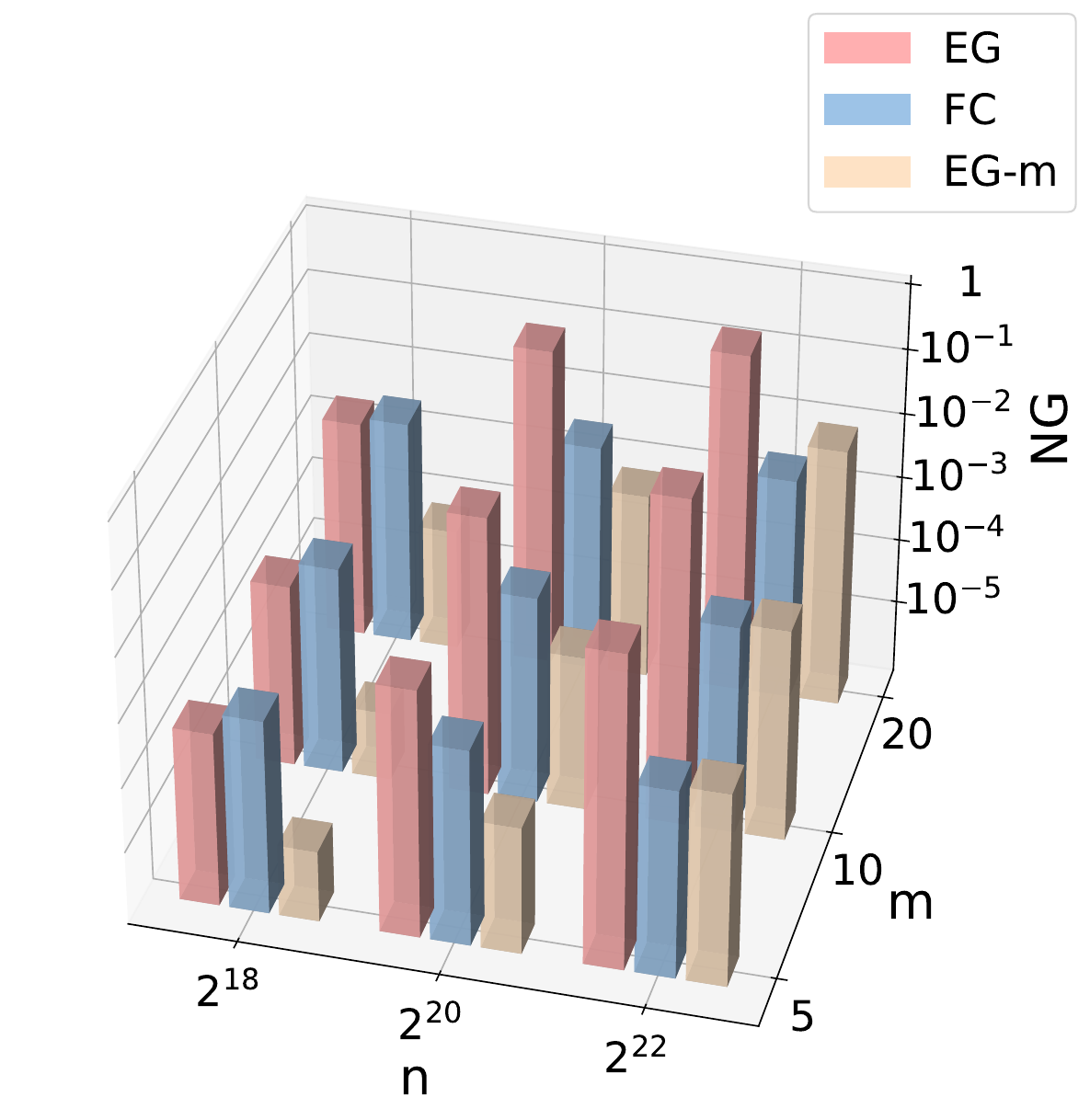}
    \label{fig:exp:FC:scale-GAP}
    \end{subfigure}
    \hspace{0.2\textwidth}
    \begin{subfigure}{0.3\textwidth}
    \includegraphics[width=\textwidth]{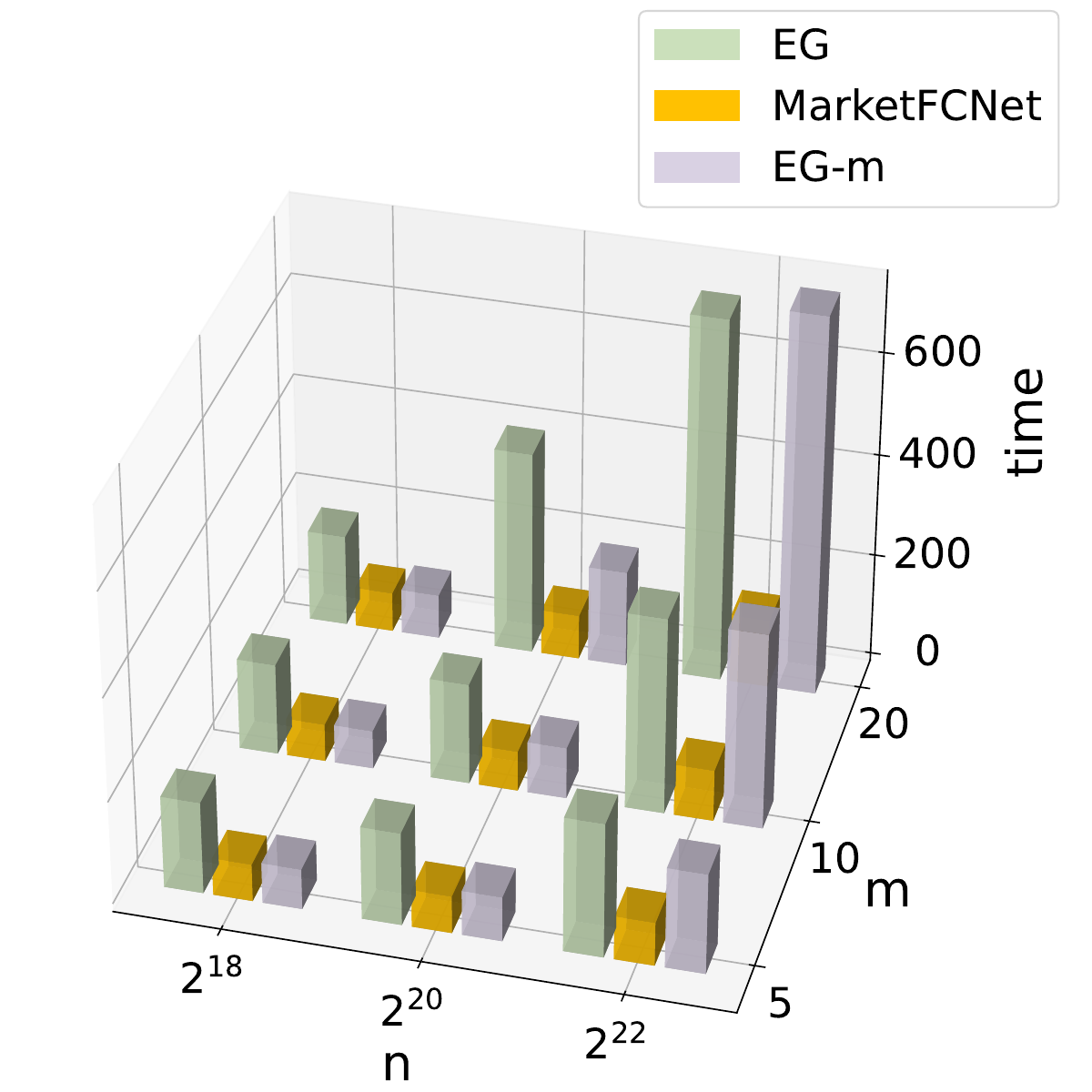}
    \label{fig:exp:FC:scale-time}
    \end{subfigure}
    \caption{The Nash Gap (left) and GPU time (right) for MarketFCNet, EG, and EG-m. Market size varies from $n=2^{18}, 2^{20}, 2^{22}$ buyers and $m=5, 10, 20$ goods.}
    \label{fig:exp:FC:scale}
\end{figure*}
\begin{figure*}[t]
    \centering

    \begin{subfigure}{0.3\textwidth}
    \includegraphics[width=\textwidth]{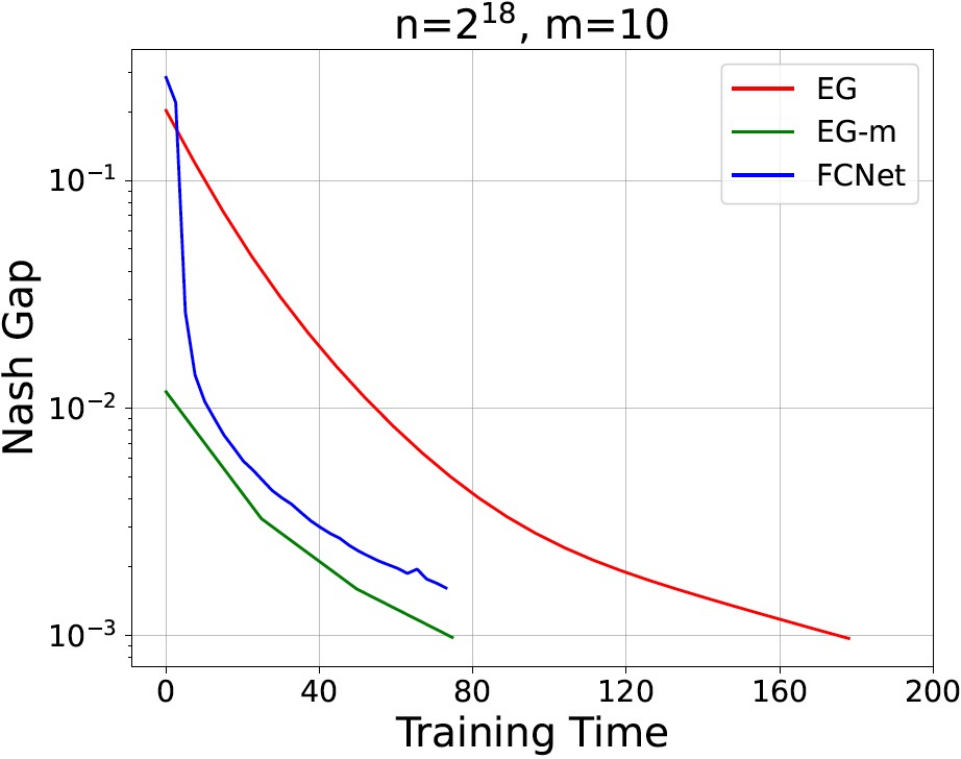}
    \caption{The curve of Nash Gap after each epoch on market with $n=2^{18}$ buyers and $m = 10$ goods.}
    \label{fig:exp:FC:loss-1}
    \end{subfigure}
    \hfill
    \begin{subfigure}{0.3\textwidth}
    \includegraphics[width=\textwidth]{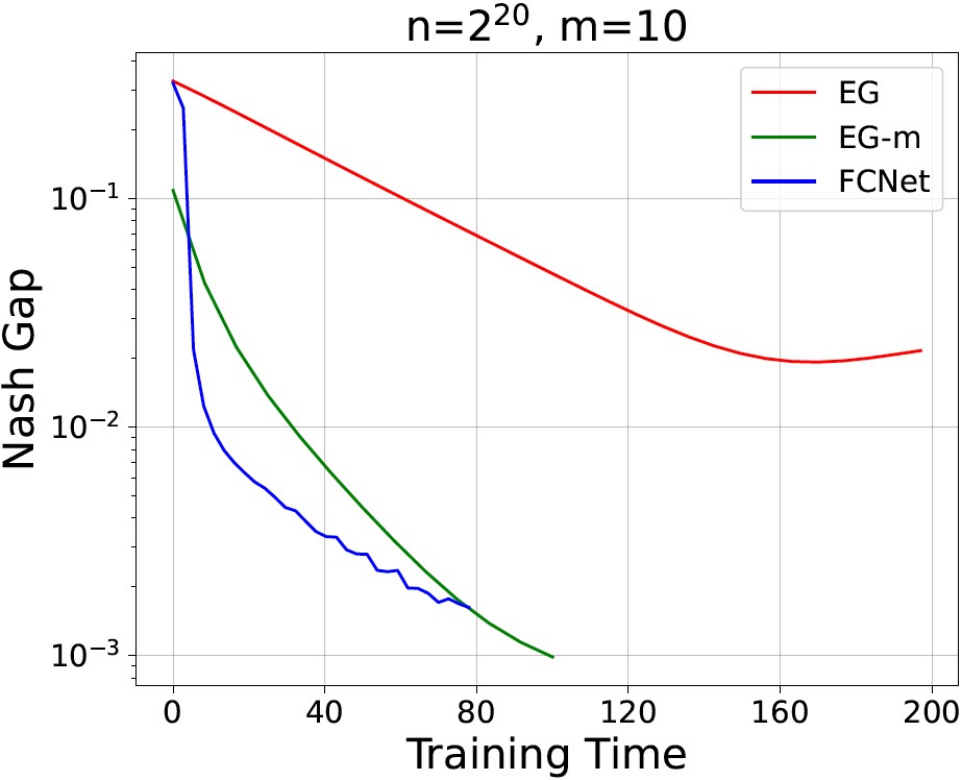}
    \caption{The curve of Nash Gap after each epoch on market with $n=2^{20}$ buyers and $m = 10$ goods.}
    \label{fig:exp:FC:loss-2}
    \end{subfigure}
    \hfill
    \begin{subfigure}{0.3\textwidth}
    \includegraphics[width=\textwidth]{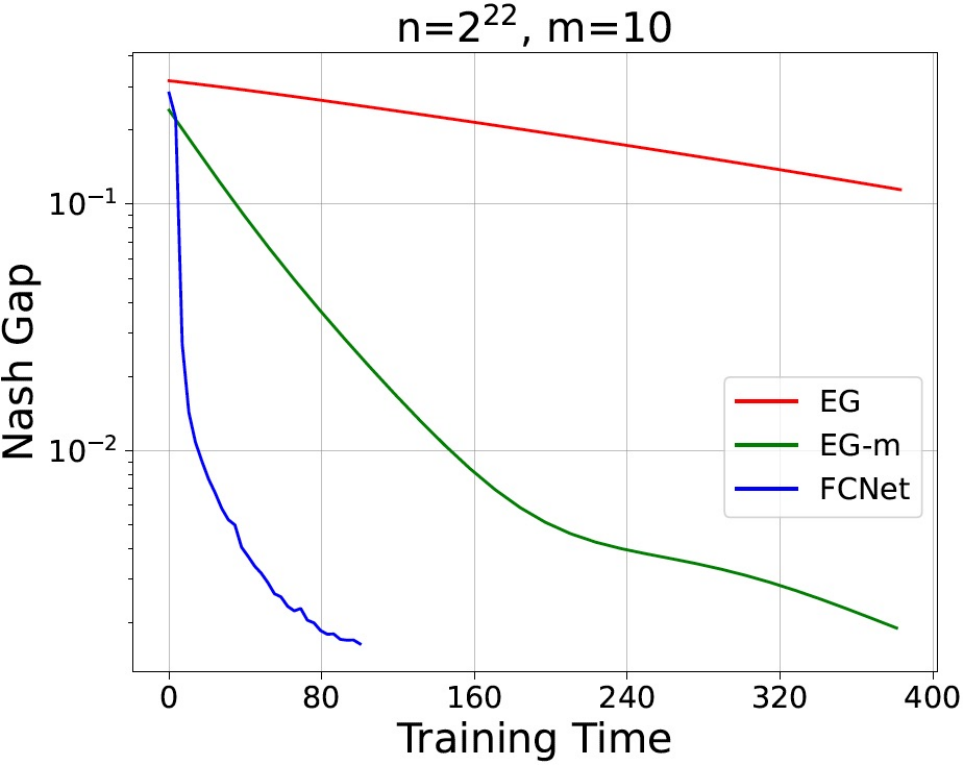}
    \caption{The curve of Nash Gap after each epoch on market with $n=2^{22}$ buyers and $m = 10$ goods.}
    \label{fig:exp:FC:loss-3}
    \end{subfigure}
    
    \caption{The curve of Nash Gap for different approaches: MarketFCNet, EG, and EG-m. Nash Gap is computed after each training epoch. Buyer size varies from $n= 2^{18},2^{20},2^{22}$, and the number of goods is fixed with $m=10$.}
    \label{fig:exp:FC:loss}
\end{figure*}

A natural question is how market size (here $n$ and $m$) will have an impact on the efficiency of MarketFCNet and traditional baselines.
We take $m=5, 10, 20$ and $n = 2^{18} = 262,114, 2^{20} = 1,048,576, 2^{22} = 4,194,304$ as the experimental settings.
For each combination of $n$ and $m$, we trained MarketFCNet and compared it with EG and EG-m. The Nash Gap and training time are provided in \cref{fig:exp:FC:scale}.
We also compute the Nash Gap after each training epoch of traditional methods and MarketFCNet, and provide the loss curve in \cref{fig:exp:FC:loss}.
Experimental results show that, as the market size varies, MarketFCNet has almost the same Nash Gap and running time. 
However, as the market size increases, both EG and EG-m have larger Nash Gaps and longer running times, demonstrating the scalability of MarketFCNet to solve large-scale contextual market equilibrium compared with traditional methods.
\section{Conclusions and Future Work}
\label{sec:conclusions}

This paper initiates the problem of large-scale contextual market equilibrium computation from a deep learning perspective. 
We believe that our approach will pioneer a promising direction for deep learning-based market equilibrium computation.

For future works, it would be promising to extend these methods to the case when only the number of goods is large , or both the numbers of goods and buyers are large, which stays a blank throughout our works. \footnote{The most ideal method is that the method shall work on the setting when the goods, or both the buyers and goods, potentially obey some certain distribution.}
Since many existing works proposed dynamics for online market equilibrium computation, it's also promising to extend our approaches to the online setting. 
Besides, both existing works and ours consider sure budgets and values for buyers, and it would be interesting to extend the fisher market and equilibrium concept when the budgets or values of buyers are stochastic or uncertain.




\newpage

\begin{credits}
\subsubsection{\ackname} 
This work is supported by Wuhan East Lake High-Tech Development Zone (also known as the Optics Valley of China, or OVC) National Comprehensive Experimental Base for Governance of Intelligent Society, and the National Natural Science Foundation of China (NSFC) under grant number [62172012]. 
The authors would like to thank Ningyuan Li, Yurong Chen, Shicheng Li, and many anonymous referees for their suggestions and help with this work.

\end{credits}

\nocite{*}

%
%
%
\printbibliography

\newpage
\onecolumn
\appendix




\section{Derivation of Fixed-price Utility for CES Utility Functions}
\label{app:derivation}

In this section we show the explicit expressions of Fixed-price Utility for CES utility functions.

We first consider the case $\alpha \ne 0,1,-\infty$.
The optimization problem for consumer $i$ is:
\begin{align}
\label{eq:opt:consumer}
    \max_{x_{ij}, \jinm}\quad & u_i(\bmx_i) = \left[ \sum_\jinm v^\alpha_{ij} x^\alpha_{ij} \right]^{1/\alpha}
    \\
    s.t.\quad& \sum_\jinm x_{ij} p_j = B_i \tag{Budget Constraint}\label{eq:budget-constraint}
    \\
    & x_{ij} \ge 0 \label{eq:spare}
\end{align}
Not hard to verify that in an optimal solution with \cref{eq:budget-constraint}, \cref{eq:spare} always holds, therefore we omit this constraint in our derivation.

We write the Lagrangian $L(\bmx_i,\lambda)$
\begin{equation}
\label{eq:opt:consumer:lag}
\begin{aligned}
    L(\bmx_i,\lambda) = u_i(\bmx_i) + \lambda (B_i - \sum_\jinm x_{ij} p_j)
\end{aligned}
\end{equation}

By $\frac{\partial L}{\partial x_{ij}} = 0$, we have
\begin{equation}
\begin{aligned}
    \frac{\partial u_i}{\partial x_{ij}^*}(\bmx_i) = \lambda p_j
\end{aligned}
\end{equation}

We derive that
\begin{align}
    \frac{\partial u_i}{\partial x_{ij}}(\bmx_i) =& \frac{1}{\alpha} \left[ \sum_\jinm v^\alpha_{ij} x^\alpha_{ij} \right]^{1/\alpha - 1} \cdot \alpha v^\alpha_{ij} x^{\alpha-1}_{ij}
    \\
    v^{\alpha}_{ij} x^{\alpha-1}_{ij} =& c p_j \qquad \cdots \text{let $c = \lambda \cdot \left[ \sum_\jinm v^\alpha_{ij} x^\alpha_{ij} \right]^{1/\alpha - 1}$}
    \\
    x_{ij}^* =& \frac{v_{ij}^{\frac{\alpha}{1-\alpha}}}{c^{\frac{1}{1-\alpha}}\cdot p_j^{\frac{1}{1-\alpha}}}
    \label{eq:allocation-expression}
\end{align}

Taking \eqref{eq:allocation-expression} into \eqref{eq:budget-constraint}, we get
\begin{align}
    B_i =& \sum_\jinm \frac{v_{ij}^{\frac{\alpha}{1-\alpha}}}{c^{\frac{1}{1-\alpha}}}\cdot p_j^{-\frac{\alpha}{1-\alpha}}
    \\
    c^{\frac{1}{1-\alpha}} =& \frac{1}{B_i} \sum_\jinm \left(\frac{v_{ij}}{p_j}\right)^{\frac{\alpha}{1-\alpha}}\label{eq:derivation:c}
\end{align}

Taking \cref{eq:derivation:c} into \cref{eq:allocation-expression}, we get
\begin{equation}
\label{eq:derivation:alloc}
    x_{ij}^* = \frac{v_{ij}^{\frac{\alpha}{1-\alpha}}}{p_j^{\frac{1}{1-\alpha}}} \cdot \frac{B_i}{c_0}
\end{equation}
where $c_0 = \sum_\jinm \left( \frac{v_{ij}}{p_j} \right)^{\frac{\alpha}{1-\alpha}}$

Taking \cref{eq:derivation:alloc} into \cref{eq:opt:consumer}, we finally have
\begin{equation}
\begin{aligned}
    u_i(\bmx_i^*) =& \left[ v_{ij}^\alpha x_{ij}^\alpha \right]^{\frac{1}{\alpha}}
    \\
    =& \left[ \sum_\jinm v_{ij}^\alpha \frac{v_{ij}^{\frac{\alpha^2}{1-\alpha}}}{p_j^{\frac{\alpha}{1-\alpha}}} c_0^\alpha \right]
    \\
    =& \left[ \sum_\jinm \left( \frac{v_{ij}}{p_j} \right)^{\frac{\alpha}{1-\alpha}} c_0^\alpha \right]
    \\
    =& B_i c_0^{\frac{1-\alpha}{\alpha}}
    \\
    \log \Tilde{u}_i(\bmp) =& \log u_i(\bmx_i^*) = \log B_i + \frac{1-\alpha}{\alpha} \log c_0
\end{aligned}
\end{equation}

For $\alpha=1$, by simple arguments we know that consumer will only buy the good that with largest value-per-cost, \ie, $v_{ij}/p_j$. Therefore, we have
\begin{align}
    \log \Tilde{u}_i(\bmp) = \log B_i + \log \max_j \frac{v_{ij}}{p_j}
\end{align}

For $\alpha=0$, we have $\log u_i(\bmx_i) = \frac{1}{v_t} \sum_\jinm v_{ij} \log x_{ij}$ where $v_t = \sum_\jinm v_{ij}$.

Similarly, we have
\begin{align}
    c p_j =& \frac{\partial \log u_i}{\partial x_{ij}} = \frac{v_{ij}}{x_{ij}}
    \\
    x_{ij}^* =& \frac{v_{ij}}{c p_j}
\end{align}

By solving budget constraints we have $c = \frac{v_t}{B_i}$, and therefore, $x_{ij}^* = \frac{v_{ij}B_i}{p_j v_t}$ and 
\begin{align}
    \log u_i(\bmx_i^*) =& \frac{1}{v_t}\sum_\jinm (v_{ij} \log \frac{v_{ij} B_i}{p_j v_t})
    \\
    =& \log B_i + \sum_\jinm \frac{v_{ij}}{v_t}\log \frac{v_{ij}}{p_j v_t}
\end{align}

For $\alpha = -\infty$, we can easily know that $v_{ij} x_{ij}^* \equiv c$ for some $c$. By solving budget constraint we have
\begin{align}
    & \sum_\jinm \frac{c p_j}{v_{ij}} = B_i
    \\
    & c = B_i \left(\sum_\jinm \frac{p_j}{v_{ij}}\right)^{-1}
    \\
    & \log \tilde{u}_i(\bmp) = \log c = \log B_i - \log \sum_\jinm \frac{p_j}{v_{ij}}
\end{align}

Above all, the log Fixed-price Utility for CES functions is
\begin{equation}
\log \tilde{u}_i(\bmp) = \begin{cases}
    \log B_i + \max_j \log \frac{v_{ij}}{p_j} \quad \text{ for $\alpha=1$}
    \\
    \log B_i + \sum_\jinm \frac{v_{ij}}{v_t}\log \frac{v_{ij}}{p_j v_t} \quad \text{ for $\alpha=0$}
    \\
    \log B_i - \log \sum_\jinm \frac{p_j}{v_{ij}} \quad \text{ for $\alpha=-\infty$}
    \\
    \log B_i + \frac{1-\alpha}{\alpha} \log c_0 \quad \text{ others}
\end{cases}
\end{equation}

\section{Omitted Proofs}
\label{app:proofs}

\subsection{Proof of \cref{prop:market-clear:equality}}

\begin{proof}
\label{prf:prop:market-clear:equality}

We consider Lagrangian multipliers $\bmp$ and use the KKT condition.
The Lagrangian becomes
\begin{align}
    L(\bmp,\bmx)=\sum_i B_i\log u_i(\bmx_i)-\sum_j p_j (\sum_i x_{ij} - Y_j)
\end{align}
and the partial derivative of $x_{ij}$ is
\begin{align}
    \frac{\partial L(\bmp,\bmx_i)}{\partial x_{ij}}=\frac{B_i}{u_i(\bmx_i)} \frac{\partial u_i}{\partial x_{ij}} - p_j
\end{align}
By complementary slackness of $x_{ij}\ge 0$, we have
\begin{align}
    \frac{B_i}{u_i(\bmx_i)} \frac{\partial u_i}{\partial x_{ij}}\le p_j \text{ for all $i$}
\end{align}
By theorem 3.1, we know that if $(\bmx,\bmp)$ is a market equilibrium, we must have $u_i(\bmx_i)>0$ for all $i$, and by condition in \cref{prop:market-clear:equality}, we can always select buyer $i$ such that $\frac{\partial u_i}{\partial x_{ij}}>0$. Therefore, we have $p_j > 0$.

As a consequence, $p_j > 0$ indicates that $\sum_j x_{ij} = V_j$ by market clearance condition.

\end{proof}

\subsection{Proof of \cref{prop:unbiased}}

\begin{proof}[Proof of \cref{prop:unbiased}]
\label{prf:prop:unbiased}

We only need to show that each of the three terms forms an unbiased estimator of the corresponding term in $\calL_\rho(\bmx_\theta;\bmll)$, then we complete the proof by the linearity of expectation.

For the first term, we have
\begin{equation}
\label{eq:unbiased-1}
\begin{aligned}
    & \bbE_{\{b^k\}_{1\le k\le 2M} \iidd \calU(\calB)}[\widehat{\OBJ}(x_\theta;b^1,...,b^{2M})] 
    \\
    =& \bbE_{\{b^k\}_{1\le k\le 2M} \iidd \calU(\calB)}\left[ \frac{1}{M} \sum_\iinM \left[ B(b^i) \log u(b^i;\bmx_\theta(b^i,\bmg)) \right] \right]
    \\
    =& \frac{1}{M} \sum_\iinM \bbE_{\{b^k\}_{1\le k\le 2M} \iidd \calU(\calB)}\left[ B(b^i) \log u(b^i;\bmx_\theta(b^i,\bmg)) \right]
    \\
    =& \frac{1}{M} \sum_\iinM \bbE_{b_i \sim \calU(\calB)}\left[ B(b^i) \log u(b^i;\bmx_\theta(b^i,\bmg)) \right]
    \\
    =& \frac{1}{M} \sum_\iinM \bbE_{b \sim \calU(\calB)}\left[ B(b) \log u(b;\bmx_\theta(b,\bmg)) \right]
    \\
    =& \bbE_{b \sim \calU(\calB)}\left[ B(b) \log u(b;\bmx_\theta(b,\bmg)) \right]
    \\
    =& \OBJ(x_\theta),
\end{aligned}
\end{equation}
where the third equation holds because $B(b^i) \log u(b^i;\bmx_\theta(b^i,\bmg)) $ is irrelevant with $b_k$ for $k\neq i$, the fourth equation holds because we replace $b^i$ with $b$.

The second term holds similarly and is thus omitted. For the last term, since $b^i$ and $b^{i+M}$ are independent, we have $x_\theta(b^i,g_j) - 1$ and $x_\theta(b^{i+M},g_j) - 1$ are independent, since they are functions of $b^i$ and $b^{i+M}$, respectively. Therefore, we have 
\begin{align*}
    & \bbE_{\{b^i, b^{i+M}\} \iidd \calU(\calB)}\left[\left(x_\theta(b^i,g_j) - 1\right)\left(x_\theta(b^{i+M},g_j) - 1\right)\right]
    \\
    =& \bbE_{b^i \sim \calU(\calB)}\left[x_\theta(b^i,g_j) - 1\right]
    \cdot \bbE_{b^{i+M} \sim \calU(\calB)}\left[x_\theta(b^{i+M},g_j) - 1\right]
    \\
    =& \left(\bbE_{b \sim \calU(\calB)}\left[x_\theta(b,g_j)\right] - 1\right)
    \cdot \left(\bbE_{b \sim \calU(\calB)}\left[x_\theta(b,g_j)\right] - 1\right)
    \\
    =& \left(\bbE_{b \sim \calU(\calB)}\left[x_\theta(b,g_j)\right] - 1\right)^2,
\end{align*}
and for the whole term, 

\begin{equation}
\label{eq:unbiased-3}
\begin{aligned}
    & \bbE_{\{b^k\}_{1\le k\le 2M} \iidd \calU(\calB)} \left[ \frac{\rho}{2}\sum_\jinm \frac{1}{M}\sum_\iinM \left(x_\theta(b^i,g_j) - 1\right)\left(x_\theta(b^{i+M},g_j) - 1\right) \right]
    \\
    =& \frac{\rho}{2M}\sum_\jinm \sum_\iinM
    \bbE_{\{b^k\}_{1\le k\le 2M} \iidd \calU(\calB)} \left[ \left(x_\theta(b^i,g_j) - 1\right)\left(x_\theta(b^{i+M},g_j) - 1\right) \right]
    \\
    =& \frac{\rho}{2M}\sum_\jinm \sum_\iinM
    \bbE_{\{b^i, b^{i+M}\} \iidd \calU(\calB)} \left[ \left(x_\theta(b^i,g_j) - 1\right)\left(x_\theta(b^{i+M},g_j) - 1\right) \right]
    \\
    =& \frac{\rho}{2M}\sum_\jinm \sum_\iinM
    \left(\bbE_{b \sim \calU(\calB)}\left[x_\theta(b,g_j)\right] - 1\right)^2
    \\
    =& \frac{\rho}{2}\sum_\jinm 
    \left(\bbE_{b \sim \calU(\calB)}\left[x_\theta(b,g_j)\right] - 1\right)^2.
\end{aligned}
\end{equation}

Combining above, we obtain that,
\begin{align*}
    \bbE_{\{b^k\}_{1\le k\le 2M} \iidd \calU(\calB)} [\hat{\calL}_\rho(x_\theta;\bmll;b^1,...,b^{2M})] = \calL_\rho(x_\theta;\bmll),
\end{align*}
which states that $\hat{\calL}_\rho(x_\theta;\bmll;b^1,...,b^{2M})$ is an unbiased estimator of $\calL_\rho(x_\theta;\bmll)$.
    
\end{proof}

\subsection{Proof of \cref{prop:price-constraint}}

Consider the market equilibrium condition $\langle \bmp^*, \bmx^*_i \rangle = B_i$, we have $\sum_j p_jx_{ij}=B_i$.
sum over this expression, we have $ \sum_i \sum_j p_jx_{ij}=\sum_i B_i$.
Then, $\sum_j p_j \sum_i x_{ij}=\sum_i B_i$.
Notice that we have $\sum_{i=1}^n x_{ij} = Y_j$ in market equilibrium, so $\sum_j p_jY_j=\sum_i B_i$, that completes the proof.

\subsection{Proof of \cref{thm:gap-largerthan0}}

\begin{proof}[Proof of \cref{thm:gap-largerthan0}]
\label{prf:thm:gap-largerthan0}
Denote $(\bmx,\bmp)$ as the market equilibrium, $\bmp$ as the price for goods and $\bmx^*_i(\bmp)$ as the optimal consumption set of buyer $i$ when the price is $\bmp$.

We have following equation:
\begin{align}
    \sum_j x_{ij} p_j =& B_i
    \\
    \bmx_i \in& \bmx^*_i(\bmp)
    \\
    \sum_\iinn x_{ij} =& Y_j
    \\
    u_i(\bmp) =& u_i(\bmx_i),\ \forall \bmp\in\bbR^m_+,\ \forall \bmx_i \in \bmx^*_i(\bmp)
\end{align}

From \cref{prop:price-constraint} we know $\sum_\iinn B_i = \sum_\jinm Y_j p_j$.

Let $\bmp'$ be some price for items such that $\sum_\jinm Y_j \bmp'_j = \sum_\iinn B_i$.
Let $\bmx'_i \in \bmx^*_i(\bmp')$ and $B'_i = \langle \bmp', \bmx_i \rangle$. We know that
\begin{align}
    \sum_\iinn B'_i = \langle \bmp', \sum_\iinn \bmx_i \rangle = \langle \bmp', \bmY\rangle = \sum_\iinn B_i 
\end{align}

For consumer $i$, $\bmx_i$ costs $B'_i$ at price $\bmp'$, thus $\frac{B_i}{B'_i}\bmx_i$ costs $B_i$ at price $\bmp'$.
Besides, $\bmx'_i$ also costs $B_i$ for price $\bmp'$, and $\bmx'$ is the optimal consumption for buyer $i$.
Then we have
\begin{align}
    u_i(\bmp') = u_i(\bmx'_i) \ge u_i(\frac{B_i}{B'_i} \bmx_i) = \frac{B_i}{B'_i} u_i(\bmx_i)
\end{align}
where the last equation is from the homogeneity(with degree 1) of utility function.

Taking logarithm and weighted sum with $B_i$, we have
\begin{align}
    \sum_\iinn B_i \log u_i(\bmp') \ge \sum_\iinn B_i \log \frac{B_i}{B'_i} + \sum_\iinn B_i \log u_i(\bmx_i)
\end{align}

Take $B_\total = \sum_\iinn B_i$, the first term in RHS becomes 
\begin{align}
    & \sum_\iinn B_i \log \frac{B_i}{B'_i}
    \\
    =& B_\total \sum_\iinn \left( \frac{B_i}{B_\total} \log \frac{B_i/B_\total}{B'_i/B_\total} \right)
    \\
    =& B_\total \cdot \text{KL}(\frac{\bm{B}}{B_\total}||\frac{\bm{B'}}{B_\total})
    \\
    \ge 0
\end{align}

Therefore,
\begin{align}
\label{prf:gap:1}
    \sum_\iinn B_i \log u_i(\bmp') \ge \sum_\iinn B_i \log u_i(\bmx_i)
\end{align}

For $\bmx'$ that satisfies market clearance, by optimality of EG program\eqref{eq:eisenberg-gale-primal}, we have
\begin{align}
\label{prf:gap:2}
    \sum_\iinn B_i \log u_i(\bmx_i) \ge \sum_\iinn B_i \log u_i(\bmx'_i)
\end{align}

\cref{prf:gap:1} and \cref{prf:gap:2} together complete the proof of the first part.

If $(\bmx,\bmp)$ constitutes a market equilibrium, it's obvious that $\LFW(\bmp)$ and $\LNW(\bmx)$ are identical, therefore $\GAP(\bmx,\bmp)=0$.

On the other hand, if $(\bmx,\bmp)$ is not a market equilibrium, but $\GAP(\bmx,\bmp)=0$, it means that the KL convergence term must equal to 0, and $B_i = B'_i$ for all $i$, which means that $\bmx_i$ costs buyer $i$ with money $B_i$ and $\bmx_i$ are in the consumption set of buyer $i$. Since $(\bmx,\bmp)$ is not a market equilibrium, there is at least one buyer that can choose a better allocation $\bmx'_i$ to improve her utility, therefore improve $\LFW(\bmp)$, and it cannot be the case that $\LFW(\bmp) = \LNW(\bmx)$, which makes a contradiction.

\end{proof}

\subsection{Proof of \cref{prop:gap-epsilon}}

We leave the formal presentation of \cref{prop:gap-epsilon} to \cref{cor:1,cor:2}
and their proofs below.

\begin{lemma}
\label{lem:gap-epsilon:1}
Assume that $u_i(\bmx_i)$ is twice differentiable and denote $H(\bmx_i)$ as the Hessian matrix of $u_i(\bmx_i)$. If following hold:
\begin{itemize}
    \item $H(\bmx_i^*)$ has rank $m-1$
    \item $||\bmx_i - \bmx_i^*|| = \varepsilon$ for some $i$
    \item $\bmx_i^*>\zeros$
\end{itemize}
then we have $\OPT - \LNW(\bmx) = \Omega(\varepsilon^2)$.
\end{lemma}

\begin{lemma}
\label{lem:gap-epsilon:2}
Denote $\Tilde{u}_i(\bmp,B_i)$ and $\bmx_i^*(\bmp,B_i)$ as the maximum utility buyer $i$ can get and the corresponding consumption for buyer $i$ when her budget is $B_i$ and prices are $\bmp$. If following hold:
\begin{itemize}
    \item $||\bmp - \bmp^*|| = \varepsilon$
    \item $\bmx_i^*(\bmp, B_i)$ is differentiable with $\bmp$.
    \item $H_X \coloneqq (\sum_\iinn \frac{\partial x_{ij}^*}{\partial p_k}(\bmp^*, B_i))_{j,k\in[m]}$ has full rank.
\end{itemize}
then we have $LFW(\bmp) - OPT = \Omega(\varepsilon^2)$.
\end{lemma}

\begin{remark}
\label{rmk:1}
It's worth notice that $H(\bmx_i^*)$ can not has full rank $m$, since $u_i(\bmx)$ is assumed to be homogeneous and thus linear in the direction $\bmx$. Therefore, we have $H(\bmx_i) \bmx_i = \zeros$ for all $\bmx_i$.

Let $C_i = \{\bmx_i\in \bbR_+^m: \langle \bmp, \bmx_i\rangle = B_i\}$ be the consumption set of buyer $i$, since $\bmx_i$ can not be parallel with $C_i$, the condition that $H(\bmx_i^*)$ has rank $m-1$ means that, $H(\bmx_i)$ is strongly concave at point $\bmx_i^*$ on the consumption set $C_i$.

Besides, we emphasize that the conditions in \cref{lem:gap-epsilon:1} and \cref{lem:gap-epsilon:2} are satisfied for CES utility with $\alpha < 1$.
\end{remark}

\begin{corollary}
\label{cor:1}
Under the assumptions in \cref{lem:gap-epsilon:1} and \cref{lem:gap-epsilon:2}, if $\GAP(\bmx,\bmp) = \varepsilon$, we have:
\begin{itemize}
    \item $||\bmp - \bmp^*|| = O(\sqrt{\varepsilon})$
    \item $||\bmx_i - \bmx_i^*|| = O(\sqrt{\varepsilon})$ for all $i$
\end{itemize}
\end{corollary}

\begin{proof}[Proof of \cref{cor:1}]
A direct inference from \cref{lem:gap-epsilon:1} and \cref{lem:gap-epsilon:2}, notice that $\GAP=\varepsilon$ indicates that $\OPT - \LNW(\bmx)\le \varepsilon$ and $\LFW(\bmp) - \OPT \le \varepsilon$.
\end{proof}

\cref{cor:1} states that, for a pair of $(\bmx,\bmp)$ that satisfy market clearance and price constraints, a small Nash Gap indicates that the point $(\bmx,\bmp)$ is close to the equilibrium point $(\bmx^*,\bmp^*)$, in the sense of Euclidean distance.

\begin{lemma}
\label{lem:gap-epsilon:3}
Assume following hold:
\begin{itemize}
    \item buyers have same utilities at $\bmx^*$, i.e. $u_i(\bmx_i^*) = u_j(\bmx_j^*) \equiv c$ for all $i,~ j$
    \item $||\bmx_i - \bmx_i^*||\le \varepsilon$ for all $i$
\end{itemize}
then, we have $|\WSW(\bmx) - \WSW(\bmx^*)| = O(\varepsilon^2)$.
\end{lemma}

\begin{remark}
\label{rmk:2}
These conditions can be held when buyers are homogeneous, \ie, $B_i = B_j$ and $u_i(\bmx)=u_j(\bmx)$ for all $i,j,\bmx\in \bbR^m_+$. Besides, consider buyers with same budgets, these conditions can also be held if the market has some ``equivariance property'', \eg, there is a $n$-cycle permutation of buyers $\rho$ and permutation of goods $\tau$, such that $u_i(\bmx_i) = u_{\rho(i)}(\tau(\bmx_{\rho(i)}))$ for all $i$ and $\tau(Y_1,...,Y_m) = (Y_1,...,Y_m)$.
\end{remark}

\begin{corollary}
\label{cor:2}
Under the assumptions in \cref{lem:gap-epsilon:1} and \cref{lem:gap-epsilon:3}, if $\GAP(\bmx,\bmp) = \varepsilon$, we have
\begin{itemize}
    \item $|\WSW(\bmx) - \WSW(\bmx^*)| = O(\varepsilon)$.
\end{itemize}
\end{corollary}

\begin{proof}
\label{prf:cor:2}
A direct inference from \cref{lem:gap-epsilon:1} and \cref{lem:gap-epsilon:3}.
\end{proof}

\subsubsection{Proof of \cref{lem:gap-epsilon:1}}
\begin{proof}[Proof of \cref{lem:gap-epsilon:1}]
\label{prf:lem:gap-epsilon:1}
We observe that
\begin{align*}
    \OPT - \LNW(\bmx) = \sum_\iinn B_i \left[\log u_i(\bmx_i^*) - \log u_i(\bmx_i)\right]
\end{align*}

Consider the Taylor expansion of $\log u_i(\bmx_i)$ and $u_i(\bmx_i)$:
\begin{align*}
    \log u_i(\bmx_i) =& \log u_i(\bmx_i^*) + \frac{1}{u_i(\bmx_i^*)}(u_i(\bmx_i) - u_i(\bmx_i^*))
    \\
    -& \frac{1}{2 u_i(\bmx_i^*)^2}(u_i(\bmx_i) - u_i(\bmx_i^*))^2
    \\
    +& O((u_i(\bmx_i) - u_i(\bmx_i^*))^3)
    \\
    u_i(\bmx_i) =& u_i(\bmx_i^*) + \frac{\partial u_i}{\partial \bmx_i}(\bmx_i^*)(\bmx_i - \bmx_i^*)
    \\
    +& \frac{1}{2}(\bmx_i - \bmx_i^*)^T H(\bmx_i^*) (\bmx_i - \bmx_i^*) + O(||\bmx_i - \bmx_i^*||^3)
\end{align*}

Notice that $||\bmx_i - \bmx_i^*||=\varepsilon$, we have
\begin{align}
    \log u_i(\bmx_i) =& \log u_i(\bmx_i^*)\notag
    \\
    +& \frac{1}{u_i(\bmx_i^*)}[ \frac{\partial u_i}{\partial x_i}(\bmx_i^*)(\bmx_i - \bmx_i^*)\cdots \varepsilon\text{ term}\label{eq:2}
    \\
    +& \frac{1}{2} (\bmx_i - \bmx_i^*)^T H(\bmx_i^*) (\bmx_i - \bmx_i^*) ]\cdots \varepsilon^2\text{ term}\label{eq:3}
    \\
    -& \frac{1}{2u_i(\bmx_i^*)^2} \left( \frac{\partial u_i}{\partial \bmx_i}(\bmx_i^*)(\bmx_i - \bmx_i^*) \right)^2\cdots \varepsilon^2\text{ term}\label{eq:4}
    \\
    +& O(\varepsilon^3)\notag
\end{align}

We next deal with \cref{eq:2} to \cref{eq:4} separately.
\paragraph{Derivation of \cref{eq:2}}
Since $\bmx_i^*$ solves the buyer $i$'s problem, we must have
\begin{equation}
    \frac{\partial u_i}{\partial x_i}(\bmx_i^*) = \lambda_i \bmp^*
    \label{eq:5}
\end{equation}
where $\lambda_i$ is the Lagrangian Multipliers for buyer $i$.

We also know that $u_i(\bmx_i)$ is homogeneous with degree 1, by Euler formula, we derive
\begin{equation}
    \langle \frac{\partial u_i}{\partial x_i}(\bmx_i), \bmx_i\rangle = u_i(\bmx_i)
    \label{eq:6}
\end{equation}

Combine \cref{eq:5} and \cref{eq:6} and take $\bmx_i = \bmx_i^*$, we derive
\begin{align*}
    \lambda_i \langle \bmp^*,\bmx_i^*\rangle =& u_i(\bmx_i^*)
    \\
    \lambda_i =& \frac{u_i(\bmx_i^*)}{B_i}
    \\
    \frac{\partial u_i}{\partial x_i}(\bmx_i^*) =& \frac{u_i(\bmx_i^*)}{B_i} \bmp^*
\end{align*}

Sum up over $i$ for \cref{eq:2}, we have
\begin{equation}
\label{eq:7}
\begin{aligned}
    & \sum_\iinn B_i \frac{1}{u_i(\bmx_i^*)} \frac{\partial u_i}{\partial x_i}(\bmx_i^*)(\bmx_i - \bmx_i^*)
    \\
    =& \bmp \sum_\iinn (\bmx_i - \bmx_i^*)
    \\
    =& 0\cdots\text{by market clearance}
\end{aligned}
\end{equation}

\paragraph{Derivation of \cref{eq:3} and \cref{eq:4}}

Combining \cref{eq:3} and \cref{eq:4}, we have
\begin{align*}
    & \frac{B_i}{2u_i(\bmx_i^*)} (\bmx_i - \bmx_i^*)^T H(\bmx_i^*) (\bmx_i - \bmx_i^*)
    - \frac{1}{2 B_i} (\bmx_i - \bmx_i^*)^T (\bmp^* \bmp^{*T}) (\bmx_i - \bmx_i^*)
    \\
    =& \frac{1}{2B_i}(\bmx_i - \bmx_i^*)^T (\frac{B_i^2}{u_i(\bmx_i^*)}H(\bmx_i^*) - \bmp^* \bmp^{*T})  (\bmx_i - \bmx_i^*)
\end{align*}

Denote $H_0(\bmx_i^*) = \frac{B_i^2}{u_i(\bmx_i^*)}H(\bmx_i^*) - \bmp^* \bmp^{*T}$, next we assert that $H_0(\bmx_i^*)$ is negative definite.

Since $H(\bmx_i^*)$ and $-\bmp^* \bmp^{*T}$ are negative semi-definite, $H_0(\bmx_i^*)$ must be negative semi-definite with $\rank(H_0(\bmx_i^*))\ge m-1$.

Let $\lambda_1 \le \lambda_2 \le \dots \le \lambda_{m-1} < \lambda_m = 0$ be eigenvalues and $v_1,...,v_n=\bmx_i^*$ be eigenvectors of $H(\bmx_i^*)$.
If $\rank(H_0(\bmx_i^*))= m-1$, it means that $\bmx_i^*$ has to be eigenvectors of $-\bmp^* \bmp^{*T}$ with eigenvalue $0$. However, we have $- \bmp^* \bmp^{*T} \bmx_i^* = - B_i \bmp^* \ne 0$, which leads to a contradiction.

Therefore, we have $\rank(H_0(\bmx_i^*))=m$ and $H_0(\bmx_i^*)$ is negative definite, we denote $\lambda^i_1\le ...,\le \lambda^i_n < 0$ as the eigenvalues of $H_0(\bmx_i^*)$, and $k$ as the universal lower bound for $|\lambda^i_n|$, then we have that,
\begin{equation}
\label{eq:8}
    \frac{1}{2}(\bmx_i - \bmx_i^*)^T H_0(\bmx_i^*) (\bmx_i - \bmx_i^*)
    \le -\frac{k}{2} \varepsilon^2
\end{equation}

By combining \cref{eq:7} and \cref{eq:8}, we have
\begin{equation}
\label{eq:9}
\begin{aligned}
    \OPT - \LNW(\bmx) =& - \sum_\iinn B_i \left[ \frac{1}{2B_i}(\bmx_i - \bmx_i^*)^T H_0(\bmx_i^*) (\bmx_i - \bmx_i^*)\right] + O(\varepsilon^3)
    \\
    \ge& \frac{k}{2}\varepsilon^2 + O(\varepsilon^3)
    \\
    =& \Omega(\varepsilon^2)
\end{aligned}
\end{equation}

\end{proof}

\subsubsection{Proof of \cref{lem:gap-epsilon:2}}

\begin{proof}[Proof of \cref{lem:gap-epsilon:2}]
\label{prf:lem:gap-epsilon:2}  

The proof is similar with \cref{prf:lem:gap-epsilon:1} by using Taylor expansion technique. Before that, we first derive some identities.

By Roy's identity, we have
\begin{align*}
    \frac{\partial \Tilde{u_i}}{\partial p_j}(\bmp,B_i) = - x_{ij}^*(\bmp, B_i) \frac{\partial \tilde{u}_i}{\partial B_i}(\bmp, B_i)
\end{align*}

Since $u(\bmx_i)$ is homogeneous with $\bmx_i$, it's easy to derive that
\begin{align*}
    \frac{\partial \tilde{u}_i}{\partial B_i}(\bmp,B_i) = \frac{\Tilde{u}_i(\bmp,B_i)}{B_i}
\end{align*}

Above all, 
\begin{align*}
    \frac{\partial \Tilde{u_i}}{\partial p_j}(\bmp,B_i) = - \frac{1}{B_i} x_{ij}^*(\bmp, B_i) \Tilde{u}_i(\bmp,B_i)
\end{align*}

Besides, 
\begin{align*}
    \frac{\partial^2 \Tilde{u_i}}{\partial p_j \partial p_k}(\bmp,B_i) =& \frac{1}{B_i^2}x_{ij}^*(\bmp, B_i)x_{ik}^*(\bmp, B_i)\Tilde{u}_i(\bmp,B_i)
    \\
    -& \frac{1}{B_i}\frac{x_{ij}^*(\bmp, B_i)}{\partial p_k}\Tilde{u}_i(\bmp,B_i)
\end{align*}

Next we consider the Taylor expansion,
\begin{align}
    \log \tilde{u}_i(\bmp) =& \log \tilde{u}_i(\bmp^*)\notag
    \\
    +& \frac{1}{\tilde{u}_i(\bmp^*)}[\frac{\partial \tilde{u}_i}{\partial \bmp}(\bmp^*) (\bmp - \bmp^*)\cdots\varepsilon \text{ term}\label{eq:10}
    \\
    +& \frac{1}{2}(\bmp - \bmp^*)^T H_p (\bmp - \bmp^*)]\cdots\varepsilon^2 \text{ term}\label{eq:11}
    \\
    -& \frac{1}{2\tilde{u}_i(\bmp^*)^2}\left[ \frac{\partial \tilde{u}_i}{\partial \bmp}(\bmp^*)(\bmp-\bmp^*) \right]^2\cdots\varepsilon^2 \text{ term}\label{eq:12}
    \\
    +& O(\varepsilon^3)\notag
\end{align}

where $H_p$ is the Hessian matrix for $\tilde{u}_i(\bmp)$.

\paragraph{Derivation of \cref{eq:10}}
We have
\begin{align*}
    & \sum_\iinn B_i \frac{1}{\tilde{u}_i(\bmp^*)} \langle \frac{\partial \tilde{u}_i}{\partial \bmp}(\bmp^*), (\bmp - \bmp^*)\rangle
    \\
    =& \sum_\iinn \langle\bmx_i^*, (\bmp - \bmp^*)\rangle
    \\
    =& \langle \ones, (\bmp - \bmp^*)\rangle\cdots\text{by market clearance}
    \\
    =& 0\cdots\text{by price constraints}
\end{align*}

\paragraph{Derivation of \cref{eq:11} and \cref{eq:12}}
These expressions become
\begin{align*}
    & \frac{1}{2\tilde{u}_i(\bmp^*)} [\frac{1}{B_i^2}\tilde{u}_i(\bmp^*)\langle \bmx_i^*, \bmp-\bmp^*\rangle^2 - \frac{1}{B_i}\tilde{u}_i(\bmp^*) (\bmp-\bmp^*)^T (\frac{\partial x_{ij}^*}{\partial p_k}(\bmp^*, B_i))_{j,k\in[m]} (\bmp-\bmp^*)]
    \\
    -& \frac{1}{2\tilde{u}_i(\bmp^*)^2}\frac{\tilde{u}_i(\bmp^*)^2}{B_i^2} \langle \bmx_i^*, \bmp-\bmp^*\rangle^2
    \\
    =& \frac{1}{2B_i} (\bmp-\bmp^*)^T (\frac{\partial x_{ij}^*}{\partial p_k}(\bmp^*, B_i))_{j,k\in[m]} (\bmp-\bmp^*)
\end{align*}

Summing up over $i$, we derive that
\begin{align*}
    \LFW(\bmp) - \OPT =& \sum_\iinn B_i \frac{1}{2B_i} (\bmp-\bmp^*)^T (\frac{\partial x_{ij}^*}{\partial p_k}(\bmp^*, B_i))_{j,k\in[m]} (\bmp-\bmp^*) + O(\varepsilon^3)
    \\
    =& \frac{1}{2}(\bmp-\bmp^*)^T H_X (\bmp-\bmp^*) + O(\varepsilon^3)
\end{align*}

Since $\bmp^*$ gets the minimum of $\LFW(\bmp)$, we must have that $H_X$ is positive semi-definite. Together with $H_X$ has full rank, we know that $H_X$ is positive definite. Denote $\lambda_m$ as the minimum eigenvalues of $H_X$, we have
\begin{align*}
    \LFW(\bmp) - \OPT \ge& \frac{\varepsilon^2 \lambda_m}{2} + O(\varepsilon^3)
    \\
    =& \Omega(\varepsilon^2)
\end{align*}

\end{proof}

\subsubsection{Proof of \cref{lem:gap-epsilon:3}}

\begin{proof}[Proof of \cref{lem:gap-epsilon:3}]
\label{prf:lem:gap-epsilon:3}
Notice that
\begin{align*}
    \WSW(\bmx) = \WSW(\bmx^*) + \sum_\iinn \langle \frac{\partial \WSW}{\partial \bmx_i}(\bmx_i^*), (\bmx_i - \bmx_i^*)\rangle + O(\varepsilon^2)
\end{align*}

We have
\begin{align*}
    & \frac{\partial \WSW}{\partial \bmx_i}(\bmx_i^*)
    \\
    =& B_i \frac{\partial u_i}{\partial \bmx_i}(\bmx_i^*)
    \\
    =& B_i \frac{u_i(\bmx_i^*)}{B_i} \bmp^*
    \\
    =& c \bmp^*
\end{align*}

Therefore,
\begin{align*}
    \WSW(\bmx) =& \WSW(\bmx^*) + \sum_\iinn c \langle \bmp^*, \bmx_i - \bmx_i^*\rangle + O(\varepsilon^2)
    \\
    =& \WSW(\bmx^*) + O(\varepsilon^2)\cdots\text{market clearance}
\end{align*}
which completes the proof.

\end{proof}

\section{Additional Experiments Details}
\label{app:experiments}

\subsection{More about baselines}
\label{app:exp:baselines}

\paragraph{EG program solver (abbreviated as EG)}
We introduce the baseline algorithm EG in this part. 
Recall the Eisenberg-Gale convex program\eqref{eq:eisenberg-gale-primal}:

\begin{equation}
\max\quad \frac{1}{n}\sum_{i=1}^n B_i \log u_i(\bmx_i)\quad {\rm s.t.} \ \frac{1}{n} \sum_{i=1}^n x_{ij} = 1, \ x\geq 0.
\end{equation}

We use the network module in pytorch to represent the parameters $\bmx \in \bbR^{n\times m}_+$, and softplus activation function to satisfy $x\ge 0$ automatedly. We use gradient ascent algorithm to optimize the parameters $\bmx$.
For constraint $\frac{1}{n}\sum_\iinn x_{ij} = 1$, we introduce Lagrangian multipliers $\lambda_j$ and minimize the Lagrangian:
\begin{align}
    \calL_\rho(\bmx;\bmlam) =& - \frac{1}{n}\sum_\iinn B_i \log u_i(\bmx_i) + \sum_\jinm \lambda_j \left( \frac{1}{n} \sum_\iinn x_{ij}  - 1 \right)
    \\
    +& \frac{\rho}{2}\sum_\jinm \left(\frac{1}{n}\sum_\iinn x_{ij} - 1 \right)^2
\end{align}

The updates of $\bmlam$ is $\lambda_j \leftarrow \lambda_j + \beta_t \rho \left( \frac{1}{n}\sum_\iinn x_{ij} - 1 \right)$, here $\beta_t$ is step size, which is identical with that in MarketFCNet.
The algorithm returns the final $(\bmx, \bmlam)$ as the approximated market equilibrium.


\paragraph{EG program solver with momentum (abbreviated as EG-m)}
The program to solve is exactly same with that in EG. The only difference is that we use gradient ascent with momentum to optimize the parameters $\bmx$.

\subsection{More Experimental Details}
\label{app:exp:details}

Without special specification, we use the experiment settings as follows. All experiments are conducted in one RTX 4090 graphics cards using 16 CPUs or 1 GPU. We set dimension of representations of buyers and goods to be $k=5$. Each elements in representation is i.i.d from $\calN(0,1)$ for normal distribution (default) contexts, $U[0,1]$ for uniform distribution contexts and $Exp(1)$ for exponential distribution contexts. Budget is generated with $B(b) = ||b||_2$, and valuation in utility function is generated with $v(b,g) = \mathrm{softplus}(\langle b,g\rangle)$, where $\mathrm{softplus}(x) = \log (1 + \exp(x))$ is a smoothing function that maps each real number to be positive. $\alpha$ in CES utility are chosen to be 0.5 by default. 
MarketFCNet is designed as a fully connected network with depth 5 and width 256 per layer. $\rho$ is chosen to be 0.2 in Augmented Lagrange Multiplier Method and the step size $\beta_t$ is chosen to be $\frac{1}{\sqrt{t}}$. We choose $K=100$ as inner iteration for each epoch, and training for $30$ epochs in MarketFCNet. 
For \emph{EG} and \emph{EG-m} baselines, we choose the inner iteration $K=1000$ when $n>1000$ and $K=100$ when $n\le 1000$ for each epoch. Baselines are enssembled with early stopping as long as $\GAP$ is lower than $10^{-3}$. Both baselines are optimized for $30$ epochs in total. 

We use Adam optimizer and learning rate $1e-4$ to optimize the allocation network in MarketFCNet. When computing $\Delta \lambda_j$ in MarketFCNet, we directly compute $\Delta \lambda_j$ rather than generate an unbiased estimator, since it does not cost too much to consider all buyers for one time.
For those baselines, we use gradient descent to optimize the parameters following existing works, and the step size is fine-tuned to be $1e+2$ for $\alpha=1$, $n>1000$; $1e+3$ for $\alpha < 1$, $n>1000$ and $1$ for $\alpha < 1,\ n\le 1000$ and $0.1$ for $\alpha = 1,\ n \le 1000$ for better performances of the baselines. Since that Lagrangian multipliers $\lambda\le 0$ will indicate an illegal Nash Gap measure, therefore, we hard code EG algorithm such that it will only return a result when it satisfies that the price $\lambda_j > 0$ for all good j. All baselines are run in GPU when $n>1000$ and CPU when $n\le 1000$.\footnote{We find in the experiments when market size is pretty large, baselines run slower on CPU than on GPU and this phenomenon reverses when market size is small. Therefore, the hardware on which baselines run depend on the market size and we always choose the faster one in experiments.}

\end{document}